\documentclass[twoside]{article}
\usepackage{amsfonts, amsthm}
\usepackage[mathcal]{eucal}
\usepackage[dvips]{graphicx}
\usepackage{subfigure}

\newtheorem{theorem}{Theorem}[section]
\newtheorem{lemma}[theorem]{Lemma}
\newtheorem{prop}[theorem]{Proposition}
\newtheorem{cor}[theorem]{Corollary}
\newtheorem{remark}[theorem]{Remark}
\newtheorem{definition}[theorem]{Definition}

\def\l{\lambda}
\def\a{\alpha}
\def\b{\beta}
\def\l{\lambda}
\def\a{\alpha}
\def\b{\beta}

\def\g{\gamma}
\def\d{\delta}
\def\e{\varepsilon}

\def\R{\mathbb{R}}

\def\C{\mathbb{C}}
\def\N{\mathbb{N}}
\def\Z{\mathbb{Z}}
\def\P{\mathbb{P}}

\def\T{\mathbb{T}}

\def\M{\mathcal{M}}
\def\Mba{\M_{\beta, \alpha}}
\def\Mmba{\M_{\beta, \alpha}}

\def\pmba{\pi_{\beta, \alpha}}
\def\W{\mathcal{W}}
\def\Wba{\W_{b,a}}
\def\P{\mathcal{P}}

\def\n{\nabla}
\def\nba{\nabla_{b,a}}
\def\cds{\!\cdot_\textbf{s}\!}
\def\mbf2{\mathbf{2}}

\def\Wa{W_1}
\def\Wb{W_2}

\begin{document}

\title{Global action-angle variables \\ for the periodic Toda lattice}
\author{Andreas Henrici \and Thomas Kappeler\footnote{Supported in part by the Swiss National Science Foundation, the programme SPECT, and the European Community through the FP6 Marie Curie RTN ENIGMA (MRTN-CT-2004-5652)}}

\maketitle

\begin{abstract}
In this paper we construct \emph{global} action-angle variables for the periodic Toda lattice.\footnote{2000 Mathematics Subject Classification: 37J35, 39A12, 39A70, 70H06}
\end{abstract}


\section{Introduction} \label{introduction}

Consider the Toda lattice with period $N$ ($N \geq 2$), 
\begin{displaymath}
\dot{q}_n = \partial_{p_n} H, \quad \dot{p}_n = -\partial_{q_n} H
\end{displaymath}
for $n \in \Z$, where the (real) coordinates $(q_n, p_n)_{n \in \Z}$ satisfy $(q_{n+N}, p_{n+N}) = (q_n, p_n)$ for any $n \in \Z$ and the Hamiltonian $H_{Toda}$ is given by
\begin{equation} \label{hamtodapq}
H_{Toda} = \frac{1}{2} \sum_{n=1}^N p_n^2 + \alpha^2 \sum_{n=1}^N e^{q_n - q_{n+1}}
\end{equation}
where $\a$ is a positive parameter, $\a > 0$. For the standard Toda lattice, $\a=1$. The Toda lattice was introduced by Toda \cite{toda} and studied extensively in the sequel. It is an important model for an integrable system of $N$ particles in one space dimension with nearest neighbor interaction and belongs to the family of lattices introduce and numerically investigated by Fermi, Pasta, and Ulam in their seminal paper \cite{fpu}. To prove the integrability of the Toda lattice, Flaschka introduced in \cite{fla1} the (noncanonical) coordinates
\begin{displaymath}
b_n := -p_n \in \R, \quad a_n := \alpha e^{\frac{1}{2} (q_n - q_{n+1})} \in \R_{>0} \quad (n \in \Z).
\end{displaymath}
These coordinates describe the motion of the Toda lattice relative to the center of mass. Note that the total momentum is conserved by the Toda flow, hence any trajectory of the center of mass is a straight line.

In these coordinates the Hamiltonian $H_{Toda}$ takes the simple form
\begin{displaymath}
H = \frac{1}{2} \sum_{n=1}^N b_n^2 + \sum_{n=1}^N a_n^2,
\end{displaymath}
and the equations of motion are
\begin{equation} \label{flaeqn}
\left\{ \begin{array}{lllll}
 \dot{b}_n & = & a_n^2 - a_{n-1}^2 \\
 \dot{a}_n & = & \frac{1}{2} a_n (b_{n+1} - b_n)
\end{array} \right. \qquad (n \in \Z).
\end{equation}
Note that $(b_{n+N}, a_{n+N}) = (b_n, a_n)$ for any $n \in \Z$, and $\prod_{n=1}^N a_n = \alpha^N$.
Hence we can identify the sequences $(b_n)_{n \in \Z}$ and $(a_n)_{n \in \Z}$ with the vectors $(b_n)_{1 \leq n \leq N} \in \R^N$ and $(a_n)_{1 \leq n \leq N} \in \R_{>0}^N$. Our aim is to study the normal form of the system of equations (\ref{flaeqn}) on the phase space
\begin{displaymath}
\M := \R^N \times \R_{>0}^N.
\end{displaymath}
This system is Hamiltonian with respect to the nonstandard Poisson structure $J \equiv J_{b,a}$, defined at a point $(b,a) = ((b_n, a_n)_{1 \leq n \leq N}$ by
\begin{equation} \label{jdef}
J = \left( \begin{array}{cc}
0 & A \\
-{}^t A & 0 \\
\end{array} \right),
\end{equation}
where $A$ is the $b$-independent $N \times N$-matrix
\begin{equation} \label{adef}
A = \frac{1}{2} \left( \begin{array}{ccccc}
a_1 & 0 & \ldots & 0 & -a_N \\
-a_1 & a_2 & 0 & \ddots & 0 \\
0 & -a_2 & a_3 & \ddots & \vdots \\
\vdots & \ddots & \ddots & \ddots & 0 \\
0 & \ldots & 0 & -a_{N-1} & a_N \\
\end{array} \right).
\end{equation}
The Poisson bracket corresponding to (\ref{jdef}) is then given by
\begin{eqnarray}
\{ F, G \}_J(b,a) & = & \langle (\nabla_b F, \nabla_a F), \, J \, (\nabla_b G, \nabla_a G) \rangle_{\R^{2N}} \nonumber\\
& = & \langle \nabla_b F, A \, \nabla_a G \rangle_{\R^N} - \langle \nabla_a F, A^t \, \nabla_b G \rangle_{\R^N}. \label{poisson}
\end{eqnarray}
where $F,G \in C^1(\M)$ and where $\nabla_b$ and $\nabla_a$ denote the gradients with respect to the $N$-vectors $b = (b_1, \ldots, b_N)$ and $a = (a_1, \ldots, a_N)$, respectively. Therefore, equations (\ref{flaeqn}) can alternatively be written as $\dot{b}_n = \{ b_n, H \}_J$, $\dot{a}_n = \{ a_n, H \}_J$ $(1 \leq n \leq N)$. Further note that
\begin{equation}
  \{ b_n, a_n \}_J = \frac{a_n}{2}; \quad \{ b_{n+1}, a_n \}_J = -\frac{a_n}{2},
\end{equation}
while $\{ b_n, a_k \}_J = 0$ for any $n,k$ with $n \notin \{ k, k+1 \}$.

Since the matrix $A$ defined by (\ref{adef}) has rank $N-1$, the Poisson structure $J$ is degenerate.
It admits the two Casimir functions\footnote{A smooth function $C: \M \to \R$ is a Casimir function for $J$ if $\{ C, \cdot \}_J \equiv 0$.}
\begin{equation} \label{casimirdef}
C_1 := -\frac{1}{N} \sum_{n=1}^N b_n \quad \textrm{and} \quad C _2 := \left( \prod_{n=1}^N a_n \right)^\frac{1}{N}
\end{equation}
whose gradients $\nba C_i \! = \! (\n_b C_i, \n_a C_i)$ ($i = 1,2$), given by
\setlength\arraycolsep{1.5pt} {
\begin{eqnarray}
\nabla_b C_1 & = & -\frac{1}{N} (1, \ldots, 1), \qquad \nabla_a C_1 = 0, \label{c1grad} \\
\nabla_b C_2 & = & 0, \qquad \nabla_a C_2 = \frac{C_2}{N} \left( \frac{1}{a_1}, \ldots, \frac{1}{a_N} \right), \label{c2grad}
\end{eqnarray}}
are linearly independent at each point $(b,a)$ of $\M$.

The main result of this paper ist the following one:
\begin{theorem} \label{fundthm}
The periodic Toda lattice admits globally defined action-angle variables. More precisely:
\begin{itemize}
\item[(i)] There exist real analytic functions $(I_n)_{1 \leq n \leq N-1}$ on $\M$ which are pairwise in involution and which Poisson commute with the Toda Hamiltonian $H$ and the two Casimir functions $C_1$, $C_2$, i.e. for any $1 \leq m,n \leq N-1$, $i = 1,2$,
  \begin{displaymath}
    \{ I_m, I_n \}_J = 0 \quad \textrm{on } \M
  \end{displaymath}
and
\begin{displaymath}
  \{ H, I_n \}_J = 0 \quad \textrm{and} \quad \{ C_i, I_n \}_J = 0 \quad \textrm{on } \M.
\end{displaymath}
\item[(ii)] For any $1 \leq n \leq N-1$ there exist a real analytic submanifold $D_n$ of codimension $2$ and a function $\theta_n: \M \setminus D_n \to \R$, defined mod $2 \pi$ and real analytic when considered mod $\pi$, so that on $\M \setminus \bigcup_{n=1}^{N-1} D_n$, $(\theta_n)_{1 \leq n \leq N-1}$ and $(I_n)_{1 \leq n \leq N-1}$ are conjugate variables. More precisely, for any $1 \leq m,n \leq N-1$, $i = 1,2$
  \begin{displaymath}
    \{ I_m, \theta_n \}_J = \delta_{mn} \quad \textrm{and} \quad \{ C_i, \theta_n \}_J = 0 \quad \textrm{on } \M \setminus D_n
  \end{displaymath}
and
\begin{displaymath}
  \{ \theta_m, \theta_n \}_J = 0 \quad \textrm{on } \M \setminus (D_m \cup D_n).
\end{displaymath}
\end{itemize}
\end{theorem}

Let $\Mba := \{ (b,a) \in \R^{2N} : (C_1, C_2) = (\beta, \alpha) \}$ denote the level set of $(C_1, C_2)$ for $(\beta, \alpha) \in \R \times \R_{> 0}$. Note that $(-\b 1_N, \a 1_N) \in \Mba$ where $1_N = (1, \ldots, 1) \in \R^N$. As the gradients $\nba C_1$ and $\nba C_2$ are linearly independent everywhere on $\M$, the sets $\Mba$ are (real analytic) submanifolds of $\M$ of codimension two. Furthermore the Poisson structure $J$, restricted to $\Mba$, becomes nondegenerate everywhere on $\Mba$ and therefore induces a symplectic structure $\nu_{\beta, \alpha}$ on $\Mba$. In this way, we obtain a symplectic foliation of $\M$ with $\Mba$ being the symplectic leaves. 

\begin{cor} \label{leafcor}
On each symplectic leaf $\Mba$, the action variables $(I_n)_{1 \leq n \leq N-1}$ are a maximal set of functionally independent integrals in involution of the periodic Toda lattice.
\end{cor}

In subsequent work \cite{ahtk2}, we will use Theorem \ref{fundthm} to construct \emph{global} Birkhoff coordinates for the periodic Toda lattice. More precisely, we introduce the model space $\P := \R^{2(N-1)} \times \R \times \R_{>0}$ endowed with the degenerate Poisson structure $J_0$ whose symplectic leaves are $\R^{2(N-1)} \times \{ \beta \} \times \{ \alpha \}$ endowed with the standard Poisson structure, and prove the following theorem:
\begin{theorem} \label{sumthm}
There exists a real analytic, canonical diffeomorphism
\begin{displaymath}
\begin{array}{ccll}
 \Omega: & (\M, J) & \to & (\P, J_0) \\
 & (b,a) & \mapsto & ((x_n, y_n)_{1 \leq n \leq N-1}, C_1, C_2)
\end{array}
\end{displaymath}
such that the coordinates $(x_n, y_n)_{1 \leq n \leq N-1}, C_1, C_2$ are global Birkhoff coordinates for the periodic Toda lattice, i.e. $(x_n, y_n)_{1 \leq n \leq N-1}$ are canonical coordinates, $C_1, C_2$ are the Casimirs and the transformed Toda Hamiltonian $\hat{H} = H \circ \Omega^{-1}$ is a function of the actions $(I_n)_{1 \leq n \leq N-1}$ and $C_1, C_2$ alone.
\end{theorem}
In \cite{ahtk3} we used Theorem \ref{sumthm} to obtain a KAM theorem for Hamiltonian perturbations of the periodic Toda lattice.

\emph{Related work:} Theorem \ref{fundthm} and Theorem \ref{sumthm} improve on earlier work on the normal form of the periodic Toda lattice in \cite{bbgk, bggk}. In particular, we construct global Birkhoff coordinates on all of $\M$ instead of a single symplectic leaf and show that techniques recently developed for treating the KdV equation (cf. \cite{kama, kapo}) and the defocusing NLS equation (cf. \cite{gkp, mcva1}) can also be applied for the Toda lattice.

\emph{Outline of the paper:} In section \ref{tools} we review the Lax pair of the periodic Toda lattice and collect some auxiliary results on the spectrum of the Jacobi matrix $L(b,a)$ associated to an element $(b,a) \in \M$. In section \ref{coord} we study the action variables $(I_n)_{1 \leq n \leq N-1}$, and in section \ref{anglecoord} we define the angle variables $(\theta_n)_{1 \leq n \leq N-1}$ on $\M \setminus \cup_{n=1}^n D_n$ using holomorphic differentials defined on the hyperelliptic Riemann surface associated to the spectrum of $L(b,a)$. In sections \ref{gradsection} and \ref{orthrelsection} we establish formulas of the gradients of the actions and angles in terms of products of fundamental solutions and prove orthogonality relations between such products which are then used in section \ref{poissonchapter} to show that $(I_n)_{1 \leq n \leq N-1}$ and $(\theta_n)_{1 \leq n \leq N-1}$ are canonical variables and to prove Theorem \ref{fundthm} and Corollary \ref{leafcor}.

\section{Preliminaries} \label{tools}

It is well known (cf. e.g. \cite{toda}) that the system (\ref{flaeqn}) admits a Lax pair formulation $\dot{L}= \frac{\partial L}{\partial t} = [B, L]$, where $L \equiv L^+(b,a)$ is the periodic Jacobi matrix defined by
\begin{equation} \label{jacobi}
L^\pm(b,a) := \left( \begin{array}{ccccc}
b_1 & a_1 & 0 & \ldots & \pm a_N \\
a_1 & b_2 & a_2 & \ddots & \vdots \\
0 & a_2 & b_3 & \ddots & 0 \\
\vdots & \ddots & \ddots & \ddots & a_{N-1} \\
\pm a_N & \ldots & 0 & a_{N-1} & b_N \\
\end{array} \right),
\end{equation}
and $B$ the skew-symmetric matrix
\begin{displaymath}
B = \left( \begin{array}{ccccc}
0 & a_1 & 0 & \ldots & -a_N \\
-a_1 & 0 & a_2 & \ddots & \vdots \\
0 & -a_2 & \ddots & \ddots & 0 \\
\vdots & \ddots & \ddots & \ddots & a_{N-1} \\
a_N & \ldots & 0 & -a_{N-1} & 0 \\
\end{array} \right).
\end{displaymath}
Hence the flow of $\dot{L} = [B, L]$ is isospectral.
\begin{prop} \label{isosp}
For a solution $\big( b(t), a(t) \big)$ of the periodic Toda lattice (\ref{flaeqn}), the eigenvalues $(\l_j^+)_{1 \leq j \leq N}$ of $L\big( b(t), a(t) \big)$ are conserved quantities.
\end{prop}

Let us now collect a few results from \cite{moer} and \cite{toda} of the spectral theory of Jacobi matrices needed in the sequel. Denote by $\M^\C$ the complexification of the phase space $\M$,
\begin{displaymath}
\M^{\C} = \{ (b,a) \in \C^{2N} : \textrm{Re }a_j > 0 \quad \forall \, 1 \leq j \leq N \}.
\end{displaymath}
For $(b,a) \in \M^\C$ we consider for any complex number $\l$ the difference equation
\begin{equation} \label{diff}
(R_{b,a} y)(k) = \l y(k) \quad (k \in \Z)
\end{equation}
where $y(\cdot) = y(k)_{k \in \Z} \in \C^{\Z}$ and $R_{b,a}$ is the difference operator
\begin{equation} \label{rdef}
R_{b,a} = a_{k-1} S^{-1} + b_k S^0 + a_k S^1
\end{equation}
with $S^m$ denoting the shift operator of order $m \in \Z$, i.e.
\begin{displaymath}
(S^m y)(k) = y(k+m) \textrm{ for } k \in \Z.
\end{displaymath}

\emph{Fundamental solutions:} The two fundamental solutions $y_1(\cdot, \l)$ and $y_2(\cdot, \l)$ of (\ref{diff}) are defined by the standard initial conditions $y_1(0, \l) = 1$, $y_1(1, \l) = 0$ and $y_2(0, \l) = 0$, $y_2(1, \l) = 1$. They satisfy the \emph{Wronskian identity}
\begin{equation} \label{wronski}
W(n) := y_1(n, \l) y_2(n+1, \l) - y_1(n+1, \l) y_2(n, \l) = \frac{a_N}{a_n}.
\end{equation}
Note that for $n = N$ one gets
\begin{equation} \label{wronskispecial}
W(N) = 1.
\end{equation}
For each $k \in \N$, $y_i(k, \l, b, a)$, $i = 1,2$, is a polynomial in $\l$ of degree at most $k-1$ and depends real analytically on $(b,a)$ (see \cite{moer}). In particular, one easily verifies that $y_2(N+1, \l, b, a)$ is a polynomial in $\l$ of degree $N$ with leading coefficient $\a^{-N}$.

\emph{Wronskian:} More generally, one defines for any two sequences $(v(n))_{n \in \Z}$ and $(w(n))_{n \in \Z}$ the Wronskian sequence $\left( W(n) \right)_{n \in \Z} = \left( W(v,w)(n) \right)_{n \in \Z}$ by
\begin{displaymath}
W(n) := v(n) w(n+1) - v(n+1) w(n).
\end{displaymath}
Let us recall the following properties of the Wronskian, which can be easily verified.
\begin{lemma}
\begin{itemize}
\item[(i)] If $y$ and $z$ are solutions of (\ref{diff}) for $\l = \l_1$ and $\l = \l_2$, respectively, then $W = W(y,z)$ satisfies for any $k \in \Z$
\begin{equation} \label{wronski1}
a_k W(k) = a_{k-1}W(k-1) + (\l_2-\l_1)y(k)z(k).
\end{equation}

\item[(ii)] If $y(\cdot, \l)$ is a $1$-parameter-family of solutions of (\ref{diff}) which is continuously differentiable with respect to the parameter $\l$ and $\dot{y}(k,\l) := \frac{\partial}{\partial\l}y(k,\l)$, then $W = W(y,\dot{y})$ satisfies for any $k \in \Z$
\begin{equation} \label{wronski2}
a_k W(k) = a_{k-1}W(k-1) + y(k,\l)^2.
\end{equation}
\end{itemize}
\end{lemma}

\emph{Discriminant:} We denote by $\Delta(\l) \equiv \Delta(\l, b, a)$ the \emph{discriminant} of (\ref{diff}), defined by
\begin{equation} \label{discrdef}
\Delta(\l) := y_1(N, \l) + y_2(N+1, \l).
\end{equation}
In the sequel, we will often write $\Delta_\l$ for $\Delta(\l)$. Note that $y_2(N+1, \l)$ is a polynomial in $\l$ of degree $N$ with leading term $\a^{-N} \l^N$, whereas $y_1(N, \l)$ is a polynomial in $\l$ of degree less than $N$, hence $\Delta(\l, b, a)$ is a polynomial in $\l$ of degree $N$ with leading term $\a^{-N} \l^N$, and it depends real analytically on $(b, a)$ (see e.g. \cite{toda}). According to Floquet's Theorem (see e.g. \cite{teschl2}), for $\l \in \C$ given, (\ref{diff}) admits a periodic or antiperiodic solution of period $N$ if the discriminant $\Delta(\l)$ satisfies $\Delta(\l) = +2$ or $\Delta(\l) = -2$, respectively. (These solutions correspond to eigenvectors of $L^+$ or $L^-$, respectively, with $L^\pm$ defined by (\ref{jacobi}).) It turns out to be more convenient to combine these two cases by considering the periodic Jacobi matrix $Q \equiv Q(b,a)$ of size $2N$ defined by
\begin{displaymath}
Q = \left( \begin{array}{cccc|cccc}
b_1 & a_1 & \ldots & 0 & 0 & \ldots & 0 & a_N \\
a_1 & b_2 & \ddots & \vdots & 0 & \ldots & & 0 \\
\vdots & \ddots & \ddots & a_{N-1} & \vdots & & & \vdots \\
0 & \ddots & \;\; a_{N-1} & b_N & a_N & \ldots & 0 & 0 \\
\hline
0 & \ldots & 0 & a_N & b_1 & a_1 & \ldots & 0 \\
0 & \ldots & & 0 & a_1 & b_2 & \ddots & \vdots \\
\vdots & & & \vdots & \vdots & \ddots & \ddots & a_{N-1} \\
a_N & \ldots & 0 & 0 & 0 & \ddots & \;\; a_{N-1} & b_N \\
\end{array} \right).
\end{displaymath}
Then the spectrum of the matrix $Q$ is the union of the spectra of the matrices $L^+$ and $L^-$ and therefore the zero set of the polynomial $\Delta^2_\l - 4$.  The function $\Delta^2_\l - 4$ is a polynomial in $\l$ of degree $2N$ and admits a product representation
\begin{equation} \label{delta2lrepr}
  \Delta^2_\l - 4 = \alpha^{-2N} \prod_{j=1}^{2N} (\l - \l_j).
\end{equation}
The factor $\a^{-2N}$ in (\ref{delta2lrepr}) comes from the above mentioned fact that the leading term of $\Delta(\l)$ is $\a^{-N} \l^N$.


\begin{figure}
  \subfigure[$N=4$]{
  \label{fig:discr4}
  \begin{minipage}[c]{0.5\linewidth}
    \centering \includegraphics[width=1.5in]{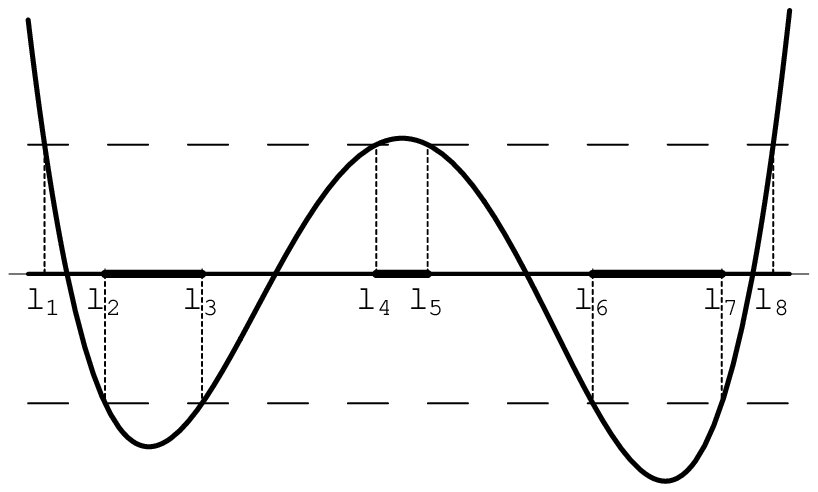}
  \end{minipage}}
  \subfigure[$N=5$]{
  \label{fig:discr5}
  \begin{minipage}[c]{0.5\linewidth}
    \centering \includegraphics[width=2in]{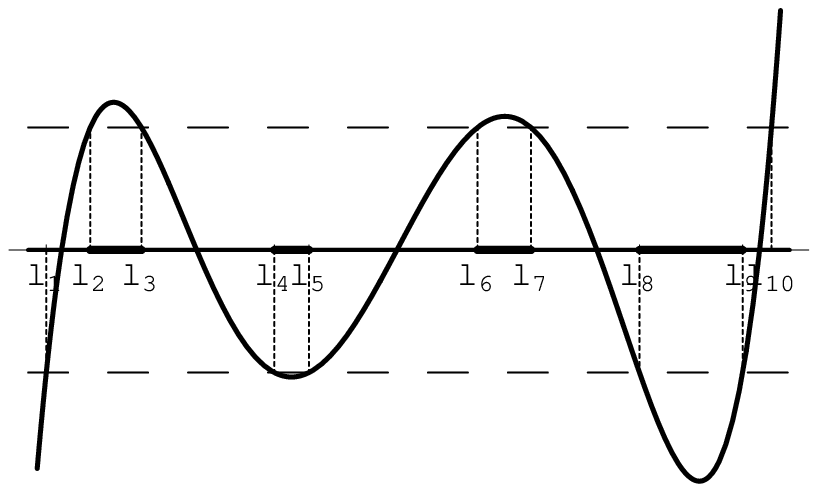}
  \end{minipage}}
  \caption{Examples of the discriminant $\Delta(\l)$}
\end{figure}

For any $(b,a) \in \M$, the matrix $Q$ is symmetric and hence the eigenvalues $(\l_j)_{1 \leq j \leq 2N}$ of $Q$ are real. When listed in increasing order and with their algebraic multiplicities, they satisfy the following relations (cf. \cite{moer})
\begin{displaymath}
\l_1 < \l_2 \leq \l_3 < \l_4 \leq \l_5 < \ldots \l_{2N-2} \leq \l_{2N-1} < \l_{2N}.
\end{displaymath}
As explained above, the $\l_j$ are periodic or antiperiodic eigenvalues of $L$ and thus eigenvalues of $L^+$ or $L^-$ according to whether $\Delta(\l) = 2$ or $\Delta(\l) = -2$. One has (cf. \cite{moer})
\begin{equation} \label{deltalambdapm2}
  \Delta(\l_1) = (-1)^N \cdot 2, \quad \Delta(\l_{2n}) = \Delta(\l_{2n+1}) = (-1)^{n+N} \cdot 2, \quad \Delta(\l_{2N}) = 2.
\end{equation}

Since $\Delta_\l$ is a polynomial of degree $N$, $\dot{\Delta}_\l \equiv \dot{\Delta}(\l) = \frac{d}{d\l} \Delta(\l)$ is a polynomial of degree $N-1$, and it admits a product representation of the form
\begin{equation} \label{dotlrepr}
  \dot{\Delta}_\l = N \a^{-N} \prod_{k=1}^{N-1} (\l - \dot{\l}_k).
\end{equation}
The zeroes $(\dot{\l}_n)_{1 \leq n \leq N-1}$ of $\dot{\Delta}_\l$ satisfy $\l_{2n} \leq \dot{\l}_n \leq \l_{2n+1}$ for any $1 \leq n \leq N-1$. The intervals $(\l_{2n}, \l_{2n+1})$ are referred to as the \emph{$n$-th spectral gap} and $\gamma_n := \l_{2n+1} - \l_{2n}$ as the \emph{$n$-th gap length}. Note that $|\Delta(\l)| > 2$ on the spectral gaps. We say that the $n$-th gap is \emph{open} if $\gamma_n > 0$ and \emph{collapsed} otherwise. The set of elements $(b,a) \in \M$ for which the $n$-th gap is collapsed is denoted by $D_n$,
\begin{equation} \label{dndef}
D_n := \{ (b,a) \in \M : \gamma_n = 0 \}.
\end{equation}
By writing the condition $\gamma_n = 0$ as $\gamma_n^2 = 0$ and exploiting the fact that $\gamma_n^2$ (unlike $\gamma_n$) is a real analytic function on $\M$, it can be shown as in \cite{kapo} that $D_n$ is a real analytic submanifold of $\M$ of codimension $2$.

\emph{Isolating neighborhoods:}
Let $(b,a) \in \M$ be given. The strict inequality $\l_{2n-1} < \l_{2n}$ guarantees the existence of a family of mutually disjoint open subsets $(U_n)_{1 \leq n \leq N-1}$ of $\C$ so that for any $1 \leq n \leq N-1$, $U_n$ is a neighborhood of the closed interval $[\l_{2n}, \l_{2n+1}]$. Such a family of neighborhoods is referred to as a family of \emph{isolating neighborhoods} (for $(b,a)$).

In the case where $(b,a) \in \M^\C$, we list the eigenvalues $(\l_j)_{1 \leq j \leq 2N}$ in lexicographic ordering\footnote{The lexicographic ordering $a \prec b$ for complex numbers $a$ and $b$ is defined by
\begin{equation} \label{lexorder}
a \prec b \quad :\Longleftrightarrow \quad \left\{
\begin{array}{l} \textrm{Re
}a < \textrm{Re }b \\
\textrm{or} \\
\textrm{Re }a = \textrm{Re }b \textrm{ and } \textrm{Im }a \leq
\textrm{Im }b. \end{array} \right.
\end{equation}
}
\begin{displaymath}
  \l_1 \prec \l_2 \prec \l_3 \prec \ldots \prec \l_{2N}.
\end{displaymath}
We then extend the gap lenghts $\gamma_n$ to all of $\M^\C$ by
\begin{displaymath}
  \gamma_n := \l_{2n+1} - \l_{2n} \quad (1 \leq n \leq N-1)
\end{displaymath}
and define
\begin{equation} \label{dncdef}
  D_n^\C := \{ (b,a) \in \W : \gamma_n = 0 \}.
\end{equation}
In the sequel, we will omit the superscript and always write $D_n$ for $D_n^\C$.

Similarly, we do this for the zeroes $(\dot{\l}_n)_{1 \leq n \leq N-1}$ of $\dot{\Delta}_\l$. The $\l_i$'s and $\dot{\l}_i$'s no longer depend continuously on $(b,a) \in \M^\C$. However, if we choose a small enough complex neighborhood $\W$ of $\M$ in $\M^\C$, then for any $(b,a) \in \W$ the closed intervals $G_n \subseteq \C$ ($1 \leq n \leq N-1$) defined by
\begin{equation} \label{gndef}
G_n := \{ (1-t) \l_{2n} + t \l_{2n+1}: 0 \leq t \leq 1 \}
\end{equation}
 are pairwise disjoint, and hence, as in the real case, there exists a family of isolating neighborhoods $(U_n)_{1 \leq n \leq N-1}$.

\begin{lemma} \label{wmungnlemma}
There exists a neighborhood $\W$ of $\M$ in $\M^\C$ such that for any $(b,a) \in \W$, there are neighborhoods $U_n$ of $G_n$ in $\C$ ($1 \leq n \leq N-1$) which are pairwise disjoint.
\end{lemma}

\begin{remark}
In the sequel, we will have to shrink the complex neighborhood $\W$ several times, but continue to denote it by the same letter.
\end{remark}

\emph{Contours $\Gamma_n$:} For any $(b,a) \in \W$ and any $1 \leq n \leq N-1$, we denote by $\Gamma_n$ a circuit in $U_n$ around $G_n$ with counterclockwise orientation.

\emph{Isospectral set:} For $(b,a) \in \M$, the set Iso$(b,a)$ of all elements $(b',a') \in \M$ so that $Q(b',a')$ has the same spectrum as $Q(b,a)$ is described with the help of the Dirichlet eigenvalues $\mu_1 < \mu_2 < \ldots < \mu_{N-1}$ of (\ref{diff}) defined by
\begin{equation} \label{mundef}
y_1(N+1, \mu_n) = 0.
\end{equation}
They coincide with the eigenvalues of the $(N-1) \times (N-1)$-matrix $L_2 = L_2(b,a)$ given by
$$ \left( \begin{array}{ccccc}
b_2 & a_2 & 0 & \ldots & 0  \\
a_2 & \ddots & \ddots & \ddots & \vdots  \\
0 & \ddots & \ddots & \ddots & 0 \\
\vdots & \ddots & \ddots & \ddots & a_{N-1} \\
0 & \ldots & 0 & a_{N-1} & b_N \\
\end{array} \right). $$
In the sequel, we will also refer to $\mu_1, \ldots, \mu_{N-1}$ as the Dirichlet eigenvalues of $L(b,a)$. Evaluating the Wronskian identity (\ref{wronski}) at $\l = \mu_n$ one sees that $\mu_n$ lies in the closure of the $n$-th spectral gap. More precisely, substituting $y_1(N+1, \mu_n) = 0$ in the identity (\ref{wronski}) with $\l = \mu_n$ yields
\begin{equation} \label{wrmu}
y_1(N,\mu_n) y_2(N+1,\mu_n) = 1.
\end{equation}
Hence the value of the discriminant at $\mu_n$ is given by
\begin{equation} \label{discdir}
\Delta(\mu_n) = y_2(N+1,\mu_n) + \frac{1}{y_2(N+1,\mu_n)}
\end{equation}
and $|\Delta(\mu_n)| \geq 2$. By Lemma \ref{speclemma} below, given the point $(b,a)$ with $b_1 = \ldots = b_N = \b$ and $a_1 = \ldots = a_N = \a$, one has $\l_{2n} = \l_{2n+1}$ and hence $\mu_n = \l_{2n}$ for any $1 \leq n \leq N-1$. It then follows from a straightforward deformation argument that $\l_{2n} \leq \mu_n \leq \l_{2n+1}$ everywhere in the real space $\M$.

Conversely, according to van Moerbeke \cite{moer}, given any (real) Jacobi matrix $Q$ with spectrum $\l_1 < \l_2 \leq \l_3 < \l_4 \leq \l_5 < \ldots \l_{2N-2} \leq \l_{2N-1} < \l_{2N}$ and any sequence $(\mu_n)_{1 \leq n \leq N-1}$ with $\l_{2n} \leq \mu_n \leq \l_{2n+1}$ for $n=1, \ldots, N-1$, there are exactly $2^r$ $N$-periodic Jacobi matrices $Q$ with spectrum $(\l_n)_{1 \leq n \leq 2N}$ and Dirichlet spectrum $(\mu_n)_{1 \leq n \leq N-1}$, where $r$ is the number of $n$'s with $\l_{2n} < \mu_n < \l_{2n+1}$.

In the case where $(b,a) \in \M^\C$, we continue to define the Dirichlet eigenvalues $(\mu_n)_{1 \leq n \leq N-1}$ by (\ref{mundef}), and we list them in lexicographic ordering $\mu_1 \prec \mu_2 \prec  \ldots \prec \mu_{N-1}$.
Then the $\mu_i$'s no longer depend continuously on $(b,a) \in \M^\C$. However, if we choose the complex neighborhood $\W$ of $\M$ in $\M^\C$ of Lemma \ref{wmungnlemma} small enough, then for any $(b,a) \in \W$ and $1 \leq n \leq N-1$, $\mu_n$ is contained in the neighborhood $U_n$ of $G_n$ (but not necessarily in $G_n$ itself).

\emph{Riemann surface $\Sigma_{b,a}$:} Denote by $\Sigma_{b,a}$ the Riemann surface obtained as the compactification of the affine curve $\mathcal{C}_{b,a}$ defined by
\begin{equation} \label{algcurve}
\{ (\l,z) \in \C^2 : z^2 = \Delta^2(\l, b, a) - 4 \}.
\end{equation}
Note that $\mathcal{C}_{b,a}$ and $\Sigma_{b,a}$ are spectral invariants. (Strictly speaking, $\Sigma_{b,a}$ is a Riemann surface only if the spectrum of $Q(b,a)$ is simple - see e.g. Appendix A in \cite{teschl2} for details in this case. If the spectrum of $Q(b,a)$ is \emph{not} simple, $\Sigma(b,a)$ becomes a Riemann surface after doubling the multiple eigenvalues - see e.g. section $2$ of \cite{kato}.)

\emph{Dirichlet divisors:} To the Dirichlet eigenvalue $\mu_n$ ($1 \leq n \leq N-1$) we associate the point $\mu_n^{*}$ on the surface $\Sigma_{b,a}$,
\begin{equation} \label{munstarred}
\mu_n^* := \left( \mu_n, \sqrt[*]{\Delta^2_{\mu_n} - 4} \right) \; \textrm{with} \;\; \sqrt[*]{\Delta^2_{\mu_n} - 4} := y_1(N, \mu_n) - y_2(N+1, \mu_n),
\end{equation}
where we used that, in view of (\ref{wrmu}) and the Wronskian identiy (\ref{wronskispecial}),
\begin{displaymath}
\Delta^2_{\mu_n} - 4 = \big( y_1(N, \mu_n) - y_2(N+1, \mu_n) \big)^2.
\end{displaymath}

\emph{Standard root:} The standard root or $s$-root for short, $\sqrt[s]{1 - \l^2}$, is defined for $\l \in \C \setminus [-1,1]$ by
\begin{equation} \label{sroot}
  \sqrt[s]{1 - \l^2} := i \l \sqrt[+]{1 - \l^{-2}}.
\end{equation}
More generally, we define for $\l \in \C \setminus \{ ta+(1-t)b \, | \, 0 \leq t \leq 1 \}$ the $s$-root of a radicand of the form $(b - \l)(\l - a)$ with $a \prec b, a \neq b$ by
\begin{equation} \label{sroot2}
\sqrt[s]{(b - \l)(\l - a)} := \frac{\gamma}{2} \sqrt[s]{1 - w^2},
\end{equation}
where $\gamma := b-a$, $\tau := \frac{b+a}{2}$ and $w := \frac{\l - \tau}{\gamma/2}$.

\emph{Canonical sheet and canonical root:} For $(b,a) \in \M$ the canonical sheet of $\Sigma_{b,a}$ is given by the set of points $(\l, \sqrt[c]{\Delta_\l^2 - 4})$ in $\mathcal{C}_{b,a}$, where the $c$-root $\sqrt[c]{\Delta_\l^2 - 4}$ is defined on $\C \setminus \bigcup_{n=0}^{N} (\l_{2n}, \l_{2n+1})$ (where $\l_0 := -\infty$ and $\l_{2N+1} := \infty$) and determined by the sign condition
\begin{equation} \label{croot}
-i \sqrt[c]{\Delta_\l^2 - 4} > 0 \quad \textrm{for} \quad \l_{2N-1} < \l < \l_{2N}.
\end{equation}
As a consequence one has for any $1 \leq n \leq N$
\begin{equation} \label{croot2}
\textrm{sign} \; \sqrt[c]{\Delta_{\l - i0}^2 - 4} = (-1)^{N+n-1} \quad \textrm{for} \quad \l_{2n} < \l < \l_{2n+1}.
\end{equation}
The definition of the canonical sheet and the $c$-root can be extended to the neighborhood $\W$ of $\M$ in $\M^\C$ of Lemma \ref{wmungnlemma}.

The $s$-root and the $c$-root will be used together in the following way: By the product representations (\ref{dotlrepr}) and (\ref{delta2lrepr}) of $\dot{\Delta}_\l$ and $\Delta^2_\l-4$, respectively, one sees that for any $(b,a)$ in $\W \setminus D_n$ with $1 \leq n \leq N-1$,
\begin{equation} \label{ddeltadelta2exp}
  \frac{\dot{\Delta}_\l}{\sqrt[c]{\Delta^2_\l-4}} = \frac{N (\l - \dot{\l}_n)}{\sqrt[s]{(\l_{2n+1} - \l)(\l - \l_{2n})}} \chi_n(\l) \quad \forall \l \in \Gamma_n
\end{equation}
where
\begin{equation} \label{chindef}
  \chi_n(\l) = \frac{(-1)^{N+n-1}}{\sqrt[+]{(\l-\l_1)(\l_{2N}-\l)}} \prod_{m \neq n} \frac{\l-\dot{\l}_m} {\sqrt[+]{(\l-\l_{2m+1})(\l-\l_{2m})}}.
\end{equation}
Note that the principal branches of the square roots in (\ref{chindef}) are well defined for $\l$ near $G_n$ and that the function $\chi_n$ is analytic and nonvanishing on $U_n$. In addition, for $(b,a)$ real, $\chi_n$ is nonnegative on the interval $(\l_{2n}, \l_{2n+1})$.

\emph{Abelian differentials:} Let $(b,a) \in \M$ and $1 \leq n \leq N-1$. Then there exists a unique polynomial $\psi_n(\l)$ of degree at most $N-2$ such that for any $1 \leq k \leq N-1$
\begin{equation} \label{psi}
\frac{1}{2\pi} \int_{c_k} \frac{\psi_n(\l)}{\sqrt{\Delta^2_\l-4}} \, d\l = \delta_{kn}.
\end{equation}
Here, for any $1 \leq k \leq N-1$, $c_k$ denotes the lift of the contour $\Gamma_k$ to the canonical sheet of $\Sigma_{b,a}$. For any $k \neq n$ with $\l_{2k} \neq \l_{2k+1}$, it follows from (\ref{psi}) that
\begin{equation} \label{psiproperty}
\frac{1}{\pi} \int_{\l_{2k}}^{\l_{2k+1}} \frac{\psi_n(\l)}{\sqrt[+]{\Delta^2_\l - 4}} \, d\l = 0.
\end{equation}
Hence in every gap $(\l_{2k}, \l_{2k+1})$ with $k \neq n$ the polynomial $\psi_n$ has a zero which we denote by $\sigma_k^n$. 
If $\l_{2k} = \l_{2k+1}$ then it follows from (\ref{psi}) and Cauchy's theorem that $\sigma_k^n = \l_{2k} = \l_{2k+1}$. As $\psi_n(\l)$ is a polynomial of degree at most $N-2$, one has
\begin{equation} \label{psiprodrepr}
  \psi_n(\l) = M_n \prod_{1 \leq k \leq N-1 \atop k \neq n} (\l - \sigma_k^n),
\end{equation}
where $M_n \equiv M_n(b,a) \neq 0$.

In a straightforward way one can prove that there exists a neighborhood $\W$ of $\M$ in $\M_{\C}$, so that for any $(b,a) \in \W$ and any $1 \leq n \leq N-1$, there is a unique polynomial $\psi_n(\l)$ of degree at most $N-2$ satisfying (\ref{psi}) for any $1 \leq k \leq N-1$ as well as the product representation (\ref{psiprodrepr}), and so that the zeroes are analytic functions on $\W$.

\begin{lemma} \label{sigmaproplemma}
Let $1 \leq n \leq N-1$ be fixed. Then the zeroes $(\dot{\l}_k)_{1 \leq k \leq N-1}$ of $\dot{\Delta}(\l)$ and $(\sigma_k^n)_{1 \leq k \leq N-1, k \neq n}$ of $\psi_n(\l)$ satisfy the estimates
\begin{eqnarray}
  \dot{\l}_k - \tau_k & = & O(\gamma_k^2), \label{dotlambdatau} \\
  \sigma_k^n - \tau_k & = & O(\gamma_k^2). \label{sigmaproperty}
\end{eqnarray}
near any given point $(b,a) \in \mathcal{W}$, where $\tau_k = \frac{1}{2}(\l_{2k+1} + \l_{2k})$.
\end{lemma}

\begin{proof}
To verify (\ref{dotlambdatau}), write $\Delta^2_\l - 4$ in the form
\begin{equation} \label{deltafactor}
\Delta^2_\l - 4 = (\l - \l_{2n}) (\l_{2n+1} - \l) p_n(\l)
\end{equation}
where $p_n$ is a polynomial which does not vanish for $\l \in U_n$. Then (\ref{dotlambdatau}) follows by differentiating (\ref{deltafactor}) with respect to $\l$ at $\dot{\l}_n$.

Fix $1 \leq k,n \leq N-1$ with $k \neq n$. In a first step we prove that $\sigma_k^n - \tau_k = O(\gamma_k)$ near any given point $(b,a) \in \W$. If $\gamma_k = 0$, then $\sigma_k^n = \tau_k$, and (\ref{sigmaproperty}) is clearly fulfilled. Hence we assume in the sequel that $\gamma_k \neq 0$. By the product formulas (\ref{psiprodrepr}) and (\ref{delta2lrepr}) for $\psi_n(\l)$ and $\Delta_\l^2 - 4$, respectively, we obtain, for $\l$ near $G_k$,
\begin{equation} \label{psideltaexp}
  \frac{\psi_n(\l)}{\sqrt[c]{\Delta^2_\l-4}} = \frac{\l - \sigma_k^n}{\sqrt[s]{(\l_{2k+1} - \l)(\l - \l_{2k})}} \zeta_k^n(\l)
\end{equation}
where
\begin{equation} \label{zetakmdef}
\zeta_k^n(\l) = \frac{M_n'(b,a)}{(\l - \sigma_n^n) \sqrt[+]{(\l - \l_1)(\l_{2N} - \l)}} \prod_{m \neq k} \frac{\l - \sigma_m^n}{\sqrt[+]{(\l_{2m+1} - \l)(\l_{2m} - \l)}},
\end{equation}
with $\sigma_n^n := \tau_n$ and $M_n'(b,a) \neq 0$. The function $\zeta_k^n$ is analytic and nonvanishing in $U_k$. Substituting (\ref{psideltaexp}) into (\ref{psi}) one gets
\begin{equation} \label{psiexp}
  \frac{1}{2\pi} \int_{\Gamma_k} \frac{\l - \sigma_k^n}{\sqrt[s]{(\l_{2k+1} - \l)(\l - \l_{2k})}} \zeta_k^n(\l) d\l = 0.
\end{equation}
We now drop the superscript $n$ for the remainder of this proof and write $\zeta_k$ as $\zeta_k(\l) = \xi_k + (\zeta_k(\l) - \xi_k)$ with $\xi_k := \zeta_k(\tau_k) \neq 0$. Note that
\begin{displaymath}
  \frac{1}{2\pi} \int_{\Gamma_k} \frac{\l - \sigma_k}{\sqrt[s]{(\l_{2k+1} - \l)(\l - \l_{2k})}} \, d\l = \tau_k - \sigma_k
\end{displaymath}
and hence (\ref{psiexp}) becomes
\begin{equation} \label{stx}
  (\sigma_k - \tau_k) \xi_k = \frac{1}{2\pi} \int_{\Gamma_k} \frac{(\l - \sigma_k) (\zeta_k(\l) - \xi_k)}{\sqrt[s]{(\l_{2k+1} - \l)(\l - \l_{2k})}} \, d\l.
\end{equation}

To estimate the integral on the right hand side of (\ref{stx}), note that
\begin{equation} \label{genestimate}
  \left| \frac{1}{2\pi} \int_{\Gamma_k} \frac{f(\l)}{\sqrt[s]{(\l_{2k+1} - \l)(\l - \l_{2k})}} \, d\l \right| \leq \max_{\l \in G_k} |f(\l)|
\end{equation}
for an arbitrary function $f$ analytic on $U_k$. We want to apply (\ref{genestimate}) for $f(\l) = (\l - \sigma_k) (\zeta_k(\l) - \xi_k)$. Note that for $\l \in G_k$,
\begin{displaymath}
  |\zeta_k(\l) - \xi_k| = |\zeta_k(\l) - \zeta_k(\tau_k)| \leq M |\gamma_k|,
\end{displaymath}
where $M = \sup \bigcup_{1 \leq k \leq N-1} \{ |\zeta_k(\l)|: \l \in G_k \}$. Hence (\ref{genestimate}) leads to
\begin{displaymath}
  |\sigma_k - \tau_k| |\xi_k| = \sup_{\l \in G_k} |\l - \sigma_k| \, O(\gamma_k).
\end{displaymath}
Dividing by $|\xi_k| \neq 0$, we get
\begin{equation} \label{sigmatau1}
  |\sigma_k - \tau_k| = \sup_{\l \in G_k} |\l - \sigma_k| \, O(\gamma_k)
\end{equation}
and in particular $|\sigma_k - \tau_k| = O(\gamma_k)$.

In a second step, we now improve the estimate (\ref{sigmatau1}). Note that
\begin{equation} \label{sigmatau2}
  \sup_{\l \in G_k} |\l - \sigma_k| \leq |\sigma_k - \tau_k| + \sup_{\l \in G_k} |\l - \tau_k| = O(\gamma_k).
\end{equation}
Combining (\ref{sigmatau1}) and (\ref{sigmatau2}), we obtain the clained estimate (\ref{sigmaproperty}).
\end{proof}

For later use, we compute the spectra of $Q(b,a)$ and $L_2(b,a)$ in the special case $(b,a) = (\b 1_N, \a 1_N)$ with $\b \in \R$ and $\a > 0$. Here $1_N$ denotes the vector $(1, \ldots, 1) \in \R^N$. These points are the equilibrium points (of the restrictions) of the Toda Hamiltonian vector field (to the symplectic leaves $\Mmba$). We compute the spectrum $(\l_j)_{1 \leq j \leq 2N}$ of the matrix $Q(\b 1_N, \a 1_N)$ and the Dirichlet eigenvalues $(\mu_k)_{1 \leq k \leq N-1}$ of $L = L(\b 1_N, \a 1_N)$ together with a normalized eigenvector $g_l = \big( g_l(j) \big)_{1 \leq j \leq N}$ of $\mu_l$, i.e. $L g_l = \mu_l g_l$, $g_l(1)=0$, and a vector $h_l = \big( h_l(j) \big)_{1 \leq j \leq N}$ which is the normalized solution of $L y = \mu_l y$ orthogonal to $g_l$ satisfying $W(h_l, g_l)(N) > 0$.
\begin{lemma} \label{speclemma}
The spectrum $(\l_j)_{1 \leq j \leq 2N}$ of $Q(\b 1_N, \a 1_N \! )$ and the Dirichlet eigenvalues $(\mu_l)_{1 \leq l \leq N-1}$ of $L(\b 1_N, \a 1_N)$ are given by
\begin{eqnarray*}
  \l_1 & = & \b - 2 \a, \\
  \l_{2l} = \l_{2l+1} = \mu_l & = & \b - 2 \a \cos \frac{l \pi}{N} \quad (1 \leq l \leq N-1), \\
  \l_{2N} & = & \b +2 \a.
\end{eqnarray*}
In particular, all spectral gaps of $Q(\b 1_N, \a 1_N)$ are collapsed. For any $1 \leq l \leq N-1$, the vectors $g_l$ and $h_l$ defined by
\begin{eqnarray}
  g_l(j) & = & (-1)^{j+1} \sqrt\frac{2}{N} \sin\frac{(j-1)l\pi}{N} \quad (1 \leq j \leq N), \label{gkformula} \\
  h_l(j) & = & (-1)^j \sqrt\frac{2}{N} \cos\frac{(j-1)l\pi}{N} \quad (1 \leq j \leq N) \label{hkformula}
\end{eqnarray}
satisfy $L y = \mu_l y$ and the normalization conditions
\begin{displaymath}
  \sum_{j=1}^N g_l(j)^2 = \sum_{j=1}^N h_l(j)^2 = 1, \quad g_l(0) > 0, \quad g_l(1) = 0;
\end{displaymath}
\begin{displaymath}
  W(h_l, g_l)(N) > 0, \quad \langle h_l, g_l \rangle_{\R^N} = 0.
\end{displaymath}
\end{lemma}

\begin{remark}
For $(b,a) = (\b 1_N, \a 1_N)$ the fundamental solutions $y_1$ and $y_2$ are given by
\begin{eqnarray}
  y_1(j,\l) & = & -\frac{\sin (\rho (j-1))}{\sin \rho} \quad (j \in \Z) \label{y1betaalpha} \\
  y_2(j,\l) & = & \frac{\sin (\rho j)}{\sin \rho} \quad (j \in \Z) \label{y2betaalpha}
\end{eqnarray}
where $\pi < \rho < 2\pi$ is determined by $\cos \rho = \frac{\l - \b}{2\a}$.
\end{remark}

\begin{proof}
For any $\l \in \R$, the difference equation (\ref{diff}) for $(\b 1_N, \a 1_N)$ reads
\begin{equation}
(R_{\b, \a} y)(k) := \b y(k) + \a y(k-1) + \a y(k+1) = \l y(k)
\end{equation}
and can be written as
\begin{equation} \label{diffeqnqba}
  y(k-1) + y(k+1) = \frac{\l - \b}{\a} y(k).
\end{equation}
Since we are looking for periodic solutions of (\ref{diffeqnqba}), we make the ansatz $y(k) = e^{\pm i \rho k}$. This leads to the characteristic equation
\begin{displaymath}
  2 \cos \rho \equiv e^{i\rho} + e^{-i\rho} = \frac{\l-\b}{\a}.
\end{displaymath}
For the solution to be $2N$-periodic, it is required that $\rho \in \frac{\pi}{N}\Z$. To put the eigenvalues in ascending order, introduce $\rho_l = (1 + \frac{l}{N}) \pi$ with $0 \leq l \leq N$. Then for any $1 \leq j \leq 2N$, there exists $0 \leq l \leq N$ such that
\begin{displaymath}
\l_j = \b + 2\a \cos \rho_l = \b - 2\a \cos \frac{l \pi}{N}.
\end{displaymath}
Note that for $l=0$, $\l_1 = \b - 2\a$ is an eigenvalue of $Q(\b 1_N, \a 1_N)$ with eigenvector $y(k) = e^{i \pi k} = (-1)^k$. Similarly, for $l=N$, $\l_{2N} = \b + 2\a$ is an eigenvalue with eigenvector $y(k) \equiv 1$. For the eigenvalue $\l_{2l} = \b - 2\a \cos \frac{l \pi}{N}$ ($1 \leq l \leq N-1$),
\begin{displaymath}
  y_\pm (k) = e^{\pm i \rho_l k}
\end{displaymath}
are two linearly independent eigenvectors. As there are $2N$ eigenvalues alltogether, $\l_{2l}$ is double for any $1 \leq l \leq N-1$, and $\l_1$ and $\l_{2N}$ are both simple. It follows that all $N-1$ gaps are collapsed and hence $\mu_l = \l_{2l}$ for all $1 \leq l \leq N-1$.

Turning to the computation of $g_k$ and $h_k$, one easily verifies that for any real number $\l \neq \pm 2\a + \b$, the fundamental solution $y_1(\cdot, \l)$ of (\ref{diffeqnqba}) with $y_1(0, \l) = 1$ and $y_1(1, \l) = 0$ is given by
\begin{displaymath}
  y_1(j, \l) = -\frac{\sin \big( \rho(j-1) \big)}{\sin \rho} \qquad (j \in \Z)
\end{displaymath}
where $\pi < \rho < 2\pi$ is determined by $\cos \rho = \frac{\l - \b}{2 \a}$, thus proving (\ref{y1betaalpha}). In the same way, one verifies (\ref{y2betaalpha}). For $\l = \mu_l = \b - 2 \a \cos \frac{l \pi}{N}$ we then get
\begin{displaymath}
\sin \big( \rho_l (j-1) \big) = \sin \big( (1 + \frac{l}{N}) \pi (j-1) \big) = (-1)^{j+1} \sin\frac{(j-1) l \pi}{N}.
\end{displaymath}
In particular, $\sin \big( \rho_l (j-1) \big) = 0$ for $j=1$ and $j=N+1$. As
\begin{displaymath}
  \sum_{j=1}^N \sin^2 \frac{(j-1) l \pi}{N} = \sum_{j=1}^N \cos^2 \frac{(j-1) l \pi}{N}
\end{displaymath}
 and these two sums add up to $N$, one sees that
\begin{equation} \label{sinsum}
\sum_{j=1}^N \sin^2 \frac{(j-1) l \pi}{N} = \frac{N}{2},
\end{equation}
yielding the claimed formula (\ref{gkformula}) for $g_l$.

By the same argument one shows that $\tilde{h}_l$ given by $(-1)^j \sqrt{\frac{2}{N}} \cos\frac{(j-1)l\pi}{N}$ (i.e. the right side of (\ref{hkformula})) satisfies $R_{\b, \a} \tilde{h}_l = \mu_l \tilde{h}_l$ and the normalization condition $\sum_{j=1}^N \tilde{h}_l(j)^2 = 1$. Using standard trigonometric identities one verifies that
\begin{displaymath}
\langle g_l, \tilde{h}_l \rangle = \sum_{j=1}^N g_l(j) \tilde{h}_l(j) = 0
\end{displaymath}
and $W(\tilde{h}_l, g_l)(N)$ can be computed to be
\begin{displaymath}
  \tilde{h}_l(N) g_l(N+1) - \tilde{h}_l(N+1) g_l(N) = -\tilde{h}_l(N+1) g_l(N) = -\tilde{h}_l(1) g_l(0) > 0.
\end{displaymath}
Hence $\tilde{h}_l$ is indeed the eigenvector with the required normalization, i.e. $h_l = \tilde{h}_l$, thus proving (\ref{hkformula}).
\end{proof}

\section{Action Variables} \label{coord}

In the next two sections, we define the candidates for action-angle variables on the phase space $\M$ and investigate some of their properties. In this section we introduce globally defined action variables $(I_n)_{1 \leq n \leq N-1}$ as proposed by Flaschka-McLaughlin \cite{flmc}.

\begin{definition} \label{actionsdef}
Let $(b,a) \in \M$. For $1 \leq n \leq N-1$,
\begin{equation} \label{actsba}
I_n := \frac{1}{2\pi} \int_{\Gamma_n} \l \frac{\dot{\Delta}_\l}{\sqrt[c]{\Delta^2_\l-4}} \; d\l
\end{equation}
where $\dot{\Delta}_\l = \frac{d}{d\l} \Delta_\l$ is the $\l$-derivative of the discriminant $\Delta_\l = \Delta(\l, b, a)$ and the contour $\Gamma_n$ and the canonical root $\sqrt[c]{\cdot}$ are given as in section \ref{tools}.
\end{definition}

\begin{remark}
The contours $\Gamma_n$ can be chosen locally independently of $(b,a)$. In view of the fact that $\Delta_\l$ is a spectral invariant of $L(b,a)$ the actions $I_n$ are entirely determined by the spectrum of $L(b,a)$. In particular, $(I_n)_{1 \leq n \leq N-1}$ are constants of motion, since by Proposition \ref{isosp}, the Toda flow is isospectral.
\end{remark}

\begin{remark}
The variables $(I_n)_{1 \leq n \leq N-1}$ can also be represented as integrals on the Riemann surface $\Sigma_{b,a}$. If $c_n$ denotes the lift of $\Gamma_n$ to the canonical sheet of $\Sigma_{b,a}$, (\ref{actsba}) becomes
\begin{equation} \label{reprcomplex}
I_n = \frac{1}{2\pi} \int_{c_n} \l \frac{\dot{\Delta}_\l}{\sqrt{\Delta^2_\l-4}} d\l \quad (1 \leq n \leq N-1).
\end{equation}
\end{remark}

From the definition (\ref{actsba}), the following result can be deduced:

\begin{prop} \label{in0ifgn0}
On the real space $\M$, each function $I_n$ is real, nonnegative, and it vanishes if $\gamma_n = 0$.
\end{prop}

\begin{proof}
Since
$$ \int_{\Gamma_n} \frac{\dot{\Delta}_\l}{\sqrt[c]{\Delta^2_\l-4}} d\l = 0, $$
it follows that
\begin{equation} \label{inrepr}
I_n = \frac{1}{2\pi} \int_{\Gamma_n} (\l-\dot{\l}_n) \frac{\dot{\Delta}_\l}{\sqrt[c]{\Delta^2_\l-4}} \, d\l.
\end{equation}
By shrinking the contour of integration to the real interval, we get
\begin{displaymath}
  I_n = \frac{1}{\pi} \int_{\l_{2n}}^{\l_{2n+1}} (-1)^{N+n-1} (\l-\dot{\l}_n) \frac{\dot{\Delta}_\l}{\sqrt[+]{\Delta^2_\l-4}} \, d\l
\end{displaymath}
by taking into account the definition (\ref{croot}) of the $c$-root. Since sign$(\l-\dot{\l}_n) \dot{\Delta}_\l = (-1)^{N+n-1}$ on $[\l_{2n}, \l_{2n+1}] \setminus \{ \dot{\l}_n \}$, the integrand is real and nonnegative, hence $I_n$ is real and nonnegative on $\M$, as claimed.

If $\gamma_n = 0$, then $\l_{2n} = \l_{2n+1}$. Hence $\dot{\l}_n = \l_{2n} = \l_{2n+1} = \tau_n$ and
\begin{displaymath}
\l - \dot{\l}_n = i \sqrt[s]{(\l_{2n+1} - \l) (\l - \l_{2n})}.
\end{displaymath}
Therefore the integrand in (\ref{actsba}) is holomorphic in the interior of the contour $\Gamma_n$, and by Cauchy's theorem the integral in (\ref{actsba}) vanishes.
\end{proof}

The action variables $(I_n)_{1 \leq n \leq N-1}$ can be extended in a straightforward way to a complex neighborhood $\W$ of $\M$ in $\M^\C$.

\begin{theorem} \label{analytic}
There exists a complex neighborhood $\W$ of $\M$ in $\M^\C$ such that for all $1 \leq n \leq N-1$, the functions $I_n$ defined by (\ref{actsba}) extend analytically to $\W$, $I_n: \W \to \C$.
\end{theorem}

\begin{proof}
Let $\W$ denote a neighborhood of $\M$ in $\M^\C$ of Lemma \ref{wmungnlemma} and define for any $1 \leq n \leq N-1$ the functions $I_n$ on $\W$ by the formula (\ref{actsba}). Let $(b,a) \in \W$ be given. Then there exists a neighborhood $\Wba$ of $(b,a)$ in $\W$ so that the integration contours $\Gamma_n$ in (\ref{actsba}) can be chosen to be the same for any element in $\Wba$ and $\dot{\Delta}_\l / \sqrt[c]{\Delta^2_\l-4}$ is analytic on $B_{\e}(\Gamma_n) \times \Wba$, where $B_{\e}(\Gamma_n) := \{ \l \in \C | \, \textrm{dist} \, (\l,\Gamma_n) < \e \}$ is the $\e$-neighborhood of $\Gamma_n$ with $\e$ sufficiently small. This shows that $I_n$ is analytic on $\W$.
\end{proof}

\begin{prop} \label{actanalytprop}
There exists a complex neighborhood $\W$ of $\M$ in $\M^\C$ such that for any $1 \leq n \leq N-1$, the quotient $I_n / \gamma_n^2$ extends analytically from $\M \setminus D_n$ to all of $\W$ and has strictly positive real part on $\W$. As a consequence, $\xi_n = \sqrt[+]{2 I_n / \gamma_n^2}$ is an analytic and nonvanishing function on $\W$, where $\sqrt[+]\cdot$ is the principal branch of the square root on $\C \setminus (-\infty, 0]$.
\end{prop}

\begin{proof}
Let $\W$ be the complex neighborhood of Theorem \ref{analytic}.
Substituting (\ref{ddeltadelta2exp}) into (\ref{inrepr}) leads to the following identity on $\W \setminus D_n$
$$ I_n = \frac{N}{2\pi} \int_{\Gamma_n} \frac{(\l-\dot{\l}_n)^2} {\sqrt[s]{(\l_{2n+1}-\l)(\l-\l_{2n})}} \chi_n(\l) d\l, $$ where $\chi_n$ is given by (\ref{chindef}).
Upon the substitution $\l(\zeta) = \tau_n + \frac{\gamma_n}{2} \zeta$, with $\tau_n = \frac{1}{2} (\l_{2n} + \l_{2n+1})$ and $\delta_n = \frac{2(\dot{\l}_n - \tau_n)}{\gamma_n}$, one then obtains
\begin{equation} \label{ingamman2}
\frac{2 I_n}{\gamma_n^2} = \frac{N}{4\pi} \int_{\Gamma_n'} \frac{(\zeta - \delta_n)^2} {\sqrt[s]{1-\zeta^2}} \chi_n(\tau_n + \frac{\gamma_n}{2} \zeta) d\zeta,
\end{equation}
where $\Gamma_n'$ is the pullback of $\Gamma_n$ under the substitution $\l = \l(\zeta)$, i.e. a circuit in $\C$ around $[-1, 1]$.
By (\ref{dotlambdatau}), $\dot{\l}_n - \tau_n = O(\gamma_n^2)$, and hence $\delta_n \to 0$ as $\gamma_n \to 0$. We conclude that
\begin{displaymath}
\lim_{\gamma_n \to 0} \frac{2 I_n}{\gamma_n^2} = \frac{N}{4\pi} \int_{\Gamma_n'} \chi_n(\tau_n) \frac{\zeta^2 \, d\zeta}{\sqrt[s]{1-\zeta^2}} = \chi_n(\tau_n) \frac{N}{2\pi} \int_{-1}^1 \frac{t^2 \, dt}{\sqrt[+]{1 - t^2}} = \frac{N}{4} \chi_n(\tau_n).
\end{displaymath}

By defining $\frac{2 I_n}{\gamma_n^2}$ by $\frac{N}{4} \chi_n(\tau_n)$ on $\W \cap D_n$, it follows that $\frac{2 I_n}{\gamma_n^2}$ is a continuous function on all of $\W$. This extended function is analytic on $\W \setminus D_n$ as is its restriction to $\W \cap D_n$. By Theorem A.6 in \cite{kapo} it then follows that $\frac{2 I_n}{\gamma_n^2}$ is analytic on all of $\W$.

By Lemma \ref{est1} below, the quotient $I_n / \gamma_n^2$ can be bounded away from zero on $\M$, $\frac{I_n}{\gamma_n^2} \geq \frac{1}{3\pi (\l_{2N}-\l_1)}$. By shrinking $\W$, if necessary, it then follows that  for any $1 \leq n \leq N-1$, the real part of $I_n / \gamma_n^2$ is positive and never vanishes on $\W$. Hence the principal branch of the square root of $2 I_n / \gamma_n^2$ is well defined on $\W$ and $\xi_n$ has the claimed properties.
\end{proof}

To show that $\sqrt[+]{\frac{2 I_n}{\gamma_n^2}}$ is well defined on $\W$, we used in the proof of Proposition \ref{actanalytprop} the following auxiliary result, which we prove in Appendix \ref{proofest}:
\begin{lemma} \label{est1}
For any $(b,a) \in \M$ and any $1 \leq n \leq N-1$,
\begin{equation} \label{estingnorig}
\gamma_n^2 \leq 3\pi (\l_{2N} - \l_1)I_n.
\end{equation}
\end{lemma}

From the definition (\ref{actsba}), Proposition \ref{in0ifgn0}, and the estimate (\ref{estingnorig}) one obtains

\begin{cor} \label{in0iffgn0}
For any $(b,a) \in \M$ and any $1 \leq n \leq N-1$,
\begin{displaymath}
  I_n = 0 \quad \textrm{if and only if} \quad \gamma_n = 0.
\end{displaymath}
\end{cor}

Actually, Lemma \ref{est1} can be improved. We finish this section with an a priori estimate of the gap lengths $\g_n$ in terms of the action variables and the value of the Casimir $C_2$ alone, which will be shown in Appendix B.

\begin{theorem} \label{ingammantotest}
For any $(b,a) \in \Mmba$ with $\b \in \R$, $\a > 0$ arbitrary,
\begin{equation} \label{estingnsumorig}
  \sum_{n=1}^{N-1} \gamma_n^2 \leq 12 \pi^2 \alpha \left( \sum_{n=1}^{N-1} I_n \right) + 9 \pi^2 (N-1) \left( \sum_{n=1}^{N-1} I_n \right)^2.
\end{equation}
\end{theorem}

\section{Angle Variables} \label{anglecoord}

In this section, we define and study the angle coordinates $(\theta_n)_{1 \leq n \leq N-1}$. Each $\theta_n$ is defined mod $2 \pi$ on $\W \setminus D_n$, where $\W$ is a complex neighborhood of $\M$ in $\M^\C$ as in Lemma \ref{wmungnlemma} and $D_n$ is given by (\ref{dncdef}).

\begin{definition} \label{angledefinition}
For any $1 \leq n \leq N-1$, the function $\theta_n$ is defined for $(b,a) \in \M \setminus D_n$ by
\begin{equation} \label{angle1}
\theta_n := \eta_n + \sum_{n \neq k=1}^{N-1} \beta_k^n \quad \textrm{mod} \; 2 \pi,
\end{equation}
where for $k \neq n$,
\begin{equation} \label{angle2}
\beta_k^n = \int_{\l_{2k}}^{\mu_k^*} \! \frac{\psi_n(\l)}{\sqrt{\Delta^2_\l-4}} \, d\l, \quad \eta_n = \int_{\l_{2n}}^{\mu_n^*} \! \frac{\psi_n(\l)}{\sqrt{\Delta^2_\l-4}} \, d\l \; (\textrm{mod} \; 2\pi),
\end{equation}
and where for $1 \leq k \leq N-1$, $\mu_k^*$ is the Dirichlet divisor defined in (\ref{munstarred}), and $\l_{2k}$ is identified with the ramification point $(\l_{2k}, 0)$ on the Riemann surface $\Sigma_{b,a}$. The integration paths on $\Sigma_{b,a}$ in (\ref{angle2}) are required to be admissible in the sense that their image under the projection $\pi: \Sigma_{b,a} \to \C$ on the first component stays inside the isolating neighborhoods $U_k$.
\end{definition}

Note that, in view of the normalization conditions (\ref{psi}) of $\psi_n$, the above restriction of the paths of integration in (\ref{angle2}) implies that $\eta_n$ and hence $\theta_n$ are well-defined mod $2 \pi$.

\begin{theorem} \label{angleanalytic}
Let $\W$ be the complex neighborhood of $\M$ in $\M^\C$ introduced in Lemma \ref{wmungnlemma}. Then for any $1 \leq n \leq N-1$, the function $\theta_n: \W \setminus D_n \to \C \, (\textrm{mod} \, \pi)$ is analytic.
\end{theorem}

\begin{remark}
As the lexicographic ordering of the eigenvalues of $Q(b,a)$ is not continuous on $\W$, it follows that $\eta_n$ and hence $\theta_n$ are only continuous mod $\pi$ on $\W$.
\end{remark}

\begin{proof}[Proof of Theorem \ref{angleanalytic}]
To see that $\theta_n: \W \setminus D_n \to \C$ (mod $\pi$) is analytic, define for any $1 \leq k \leq N-1$ the set
\begin{displaymath}
E_k := \{ (b,a) \in \M^\C: \mu_k(b,a) \in \{ \l_{2k}(b,a), \l_{2k+1}(b,a) \} \}.
\end{displaymath}
Below, we show that for any $1 \leq k \leq N-1$ with $k \neq n$, $\beta_k^n$ is analytic on $\W \setminus (D_k \cup E_k)$, that its restrictions to $D_k \cap \W$ and $E_k \cap \W$ are weakly analytic\footnote{Let $E$ and $F$ be complex Banach spaces, and let $U \subset E$ be open. The map $f: U \to F$ is \emph{weakly analytic on $U$}, if for each $u \in U$, $h \in E$ and $L \in F^*$, the function $z \mapsto Lf(u + zh)$ is analytic in some neighborhood of the origin in $\C$.}, and that it is continuous on $\W$. Together with the fact that $E_k \cap \W$ and $D_k \cap \W$ are analytic subvarieties of $\W$ it then follows that $\beta_k^n$ is analytic on $\W$ - see Theorem A.6 in \cite{kapo}. Similar results can be shown for $\beta_n^n = \eta_n$ (mod $\pi$) on $\W \setminus D_n$, and one concludes that $\theta_n$ (mod $\pi$) is analytic on $\W \setminus D_n$.

To prove that $\beta_k^n$, $k \neq n$, is analytic on $\W \setminus (D_k \cup E_k)$, note that since $\l_{2k}$ is a simple eigenvalue on $\W \setminus D_k$, it is analytic there. Furthermore, $\mu_k^*$ is an analytic function on the (sufficiently small) neighborhood $\W$ of $\M$ in $\M^\C$. On $\W \setminus (D_k \cup E_k)$ we can use the substitution $\l = \l_{2k} + z$ to get
$$ \beta_k^n = \int_{\l_{2k}}^{\mu_k^*} \frac{\psi_n(\l)}{\sqrt{\Delta^2_\l - 4}} d\l = \int_0^{\mu_k^* - \l_{2k}} \frac{\psi_n(\l_{2k} + z)}{\sqrt{z} \sqrt{D(z)}} dz, $$
where $D(z) = \frac{\Delta^2(\l_{2k} + z) - 4}{z}$ is analytic near $z = 0$ and $D(0) \neq 0$. Note that $D(z)$ does not vanish for $z$ on an admissible integration path not going through $\l_{2k+1}$. Such a path exists since $(b,a)$ is in the complement of $E_k$. Furthermore $\psi_n(\l_{2k} + z)$ and $D(z)$ are analytic in $z$ near such a path and depend analytically on $(b,a) \in \W \setminus (D_k \cup E_k)$. Combining these arguments shows that $\b_k^n$ is analytic on $\W \setminus (D_k \cup E_k)$.

For $k \neq n$ with $\l_{2k} \neq \l_{2k+1}$ one has
\begin{equation} \label{analemma}
\int_{\l_{2k}}^{\l_{2k+1}} \frac{\psi_n(\l)}{\sqrt{\Delta^2_\l - 4}} d\l = 0.
\end{equation}
As $\sigma_k^n = \l_{2k}$ if $\l_{2k} = \l_{2k+1}$ one sees that (\ref{analemma}) continues to hold for $(b,a) \in E_k \cap W$ with $\l_{2k} = \l_{2k+1}$ and we have $\beta_k^n |_{E_k \cap \W} \equiv 0$. To prove the analyticity of $\beta_k^n |_{D_k \cap \W}$ consider the representation (\ref{psideltaexp}) of $\frac{\psi_n(\l)}{\sqrt{\Delta^2_\l - 4}}$.
For $(b,a) \in D_k \cap \W$, one has
\begin{displaymath}
  \l_{2k} = \l_{2k+1} = \tau_k = \sigma_k^n,
\end{displaymath}
which implies that the factor $\frac{\l - \sigma_k^n}{\sqrt[s]{(\l_{2k+1} - \l)(\l - \l_{2k})}}$ in (\ref{psideltaexp}) equals $\pm i$. Hence we can write
\begin{displaymath}
\beta_k^n = \int_{\l_{2k}}^{\mu_k^*} \frac{\psi_n(\l)}{\sqrt{\Delta^2_\l-4}} \; d\l = \pm i \int_{\tau_k}^{\mu_k} \zeta_k^n(\l) d\l.
\end{displaymath}
As $\mu_k$ is analytic on $\W$, it then follows that $\b_k^n |_{D_k \cap \W}$ is analytic. To see that $\beta_k^n$ is continuous on $\W$, one separately shows that $\beta_k^n$ is continuous at points in $\W \setminus (D_k \cup E_k)$, $E_k \cap \W \setminus D_k$, $D_k \cap \W \setminus E_k$ and $D_k \cap E_k \cap \W$, where for the proof of the continuity of $\b_k^n$ at points in $D_k \cap E_k \cap \W$ we use (\ref{psideltaexp}) and the estimate $\sigma_k^n - \tau_k = O(\gamma_k^2)$ of Lemma \ref{sigmaproplemma}.

By (\ref{analemma}), $\eta_n$ vanishes mod $\pi$ on $E_n \cap \W \setminus D_n$. Arguing in a similar way as for $\beta_k^n$ one then concludes that $\eta_n$ (mod $\pi$) is analytic on $\W \setminus D_n$.
\end{proof}

\section{Gradients} \label{gradsection}

In this section we establish formulas of the gradients of $I_n$, $\theta_n$ ($1 \leq n \leq N-1$) on $\M$ in terms of products of the fundamental solutions $y_1$ and $y_2$.

Consider the discriminant for a fixed value of $\l$ as a function on $\M$,
$$ \Delta_\l(b,a) = y_1(N) + y_2(N+1). $$
Then $\Delta_\l$ is a real analytic function on $\M$. To obtain a formula for the gradients of $y_1(N)$ and $y_2(N+1)$ with respect to $b$, differentiate $R_{b,a} y_i = \l y_i$ with respect to $b$ in the direction $v \in \R^N$ to get
\begin{equation} \label{diffinhomb}
(R_{b,a} - \l)  \langle \n_b y_i, v \rangle (k) = - v_k y_i(k).
\end{equation}
Differentiating $R_{b,a} y_i = \l y_i$ with respect to $a$ in the direction $u \in \R^N$ leads to
\begin{equation} \label{diffinhoma}
(R_{b,a} - \l) \langle \n_a y_i, u \rangle (k) = - u_{k-1} y_i(k-1) - u_k y_i(k+1).
\end{equation}
Taking the sum of (\ref{diffinhomb}) and (\ref{diffinhoma}) yields
\begin{equation} \label{diffinhomprov}
  (R_{b,a} - \l) \left( \langle \n_b y_i, v \rangle + \langle \n_a y_i, u \rangle \right) (k) = -(R_{v,u} y_i)(k)
\end{equation}
which we can rewrite as
\begin{equation} \label{diffinhom}
(R_{b,a} - \l) \langle \n_{b,a} y_i, (v,u) \rangle (k) = - (R_{v,u} y_i)(k),
\end{equation}
where $\langle \cdot, \cdot \rangle$ in (\ref{diffinhom}) now denotes the standard scalar product in $\R^{2N}$, whereas in (\ref{diffinhomb}), (\ref{diffinhoma}), and (\ref{diffinhomprov}) it is the one in $\R^N$. The inhomogeneous Jacobi difference equation (\ref{diffinhom}) for the sequence $\langle \n_{b,a} y_i, (v,u) \rangle (k)$ can be integrated using the discrete analogue of the method of the variation of constants used for inhomogeneous differential equations. As $\langle \n_{b,a} y_i, (v,u) \rangle (0) = \langle\n_{b,a} y_i, (v,u) \rangle (1) = 0$, one obtains in this way for $m \geq 1$

\begin{eqnarray}
\langle \n_{b,a} \, y_i, (v,u) \rangle (m) & = & -\Bigg( \frac{y_{2}(m)}{a_N} \sum_{k=1}^{m} y_1(k)(R_{v,u}y_i)(k) \nonumber\\
&& \quad \quad - \frac{y_1(m)}{a_N} \sum_{k=1}^{m} y_2(k)(R_{v,u}y_i)(k) \Bigg). \label{diffint}
\end{eqnarray}

In the sequel, we will use (\ref{diffint}) to derive various formulas for the gradients. The common feature among these formulas is that they involve products between the fundamental solutions $y_1$ and $y_2$ of (\ref{diff}). Whereas the gradients with respect to $b = (b_1, \ldots, b_N)$ involve products computed by componentwise multiplication, the gradients with respect to $a = (a_1, \ldots, a_N)$ involve products obtained by multiplying shifted components, reflecting the fact that the $b_j$ are the diagonal elements of the symmetric matrix $L(b,a)$, whereas the $a_j$ are the off-diagonal elements of $L(b,a)$.

To simplify notation for the formulas in this section, we define for sequences $\big( v(j)_{j \in \Z} \big), \big( w(j)_{j \in \Z} \big) \subseteq \C$ the $N$-vectors
\begin{eqnarray}
v \cdot w & := & \left( v(j) w(j) \right)_{1 \leq j \leq N}, \label{nprod1}\\
v \cdot S  w & := & \left( v(j) w(j+1) \right)_{1 \leq j \leq N}, \label{nprod2}
\end{eqnarray}
where $S$ denotes the shift operator of order $1$. 
Combining (\ref{nprod1}) and (\ref{nprod2}), we define the $2N$-vector
\begin{equation} \label{vector2n}
v \cds w := (v \cdot w, v \cdot Sw + w \cdot Sv).
\end{equation}
In case $v = w$ we also use the shorter notation
\begin{equation} \label{vector2nvv}
v^\mbf2 := v \cds v.
\end{equation}
Written componentwise, $v \cds w$ is the $2N$-vector
\setlength\arraycolsep{0.1pt} {
\begin{displaymath} \!\!\!\!\!\ (v \cds w)(j) \! = \!\!
\left\{ \begin{array}{cc}
v(j) w(j) & \quad (1 \leq j \leq N) \\
v(j \!-\! N) w(j \!-\! N \!+\! 1) + v(j \!-\! N \!+\! 1) w(j \!-\! N) & \quad (N \! < \! j \leq \! 2N) \\
\end{array} \right.
\end{displaymath}}

\begin{prop} \label{discrprod}
For any $(b,a) \in \M$, the gradient $\n_{b,a} \Delta_\l = (\n_b \Delta_\l, \n_a \Delta_\l)$ is given by
\setlength\arraycolsep{2pt} {
\begin{eqnarray}
-a_N \n_b \Delta_\l & = & y_2(N) \, y_1 \! \cdot \! y_1 - y_1(N\!\!+\!\!1) \, y_2 \! \cdot \! y_2 + \big( y_2(N\!\!+\!\!1) - y_1(N) \big) \, y_1 \! \cdot \! y_2 \label{nbdelta1} \\
-a_N \n_a \Delta_\l & = & 2 y_2(N) y_1 \cdot S y_1 - 2 y_1(N+1) y_2 \cdot S y_2 \nonumber\\
&& \qquad + \big( y_2(N+1) - y_1(N) \big) \big (y_1 \cdot S y_2 + y_2 \cdot S y_1 \big) \label{nadelta1}
\end{eqnarray}}
or in the notation introduced above
\begin{equation} \label{dpr}
\n_{b,a} \Delta_\l = -\frac{1}{a_N} \left( y_2(N) y_1^\mbf2 - y_1(N+1) y_2^\mbf2  + \left( y_2(N+1) - y_1(N) \right) y_1 \cds y_2 \right).
\end{equation}
The gradients $\n_b \Delta_\l$ and $\n_a \Delta_\l$ admit the representations $(1 \leq m \leq N)$
\setlength\arraycolsep{0.2pt} {
\begin{eqnarray}
\frac{\partial \Delta_\l}{\partial b_m} & = & -\frac{1}{a_m} y_2(N, \l, S^m b, S^m a), \label{nbdeltashift} \\
\frac{\partial \Delta_\l}{\partial a_m} & = & - \! \left( \frac{1}{a_m} y_2(N\!\!+\!\!1, \l, S^m b, S^m a) \!+\! \frac{1}{a_{m\!+\!1}} y_2(N\!\!-\!\!1, \l, S^{m+1} b, S^{m+1} a) \! \right) \! . \label{nadeltashift}
\end{eqnarray}}
\end{prop}

\begin{proof}
The claimed formula (\ref{dpr}) follows from the defintion of $\Delta_\l$ and formula (\ref{diffint}). Indeed, evaluate (\ref{diffint}) for $i=1$ and $m=N$ to get
\setlength\arraycolsep{1pt} {
\begin{eqnarray}
\langle \n_{b,a} y_1, (v,\!u) \rangle (N) & \! = \! & -\frac{y_{2}(N)}{a_N} \! \sum_{k=1}^{N} y_1(k) \left( u_{k-1} y_1(k-1) + v_k y_1(k) + u_{k} y_1(k+1) \right) \nonumber\\
&& \!\!\!\!\!\!\!\!\!\!\!\!\!\!\!\!\!\!\!\!\!\!\!\!\!\!\! + \frac{y_1(N)}{a_N} \! \sum_{k=1}^{N} y_2(k) \left( u_{k-1} y_1(k-1) + v_k y_1(k) + u_{k} y_1(k+1) \right). \label{dbay12zeile}
\end{eqnarray}}
In order to identify these two sums with $\langle y_1^\mbf2, (v,u) \rangle$ and $\langle y_1 \cds y_2, (v,u) \rangle$, respectively, note that
$$ \sum_{k=1}^N u_{k-1} y_1(k) y_1(k-1) = \sum_{k=1}^N u_k y_1(k) y_1(k+1) + u_N T_1 $$
where
\begin{displaymath}
  T_1 := y_1(0) y_1(1) - y_1(N) y_1(N+1).
\end{displaymath}
For the second sum in (\ref{dbay12zeile}), we get an expression of the same type with a similar correction term
\begin{displaymath}
  T_2 := y_1(0) y_2(1) - y_1(N) y_2(N+1).
\end{displaymath}
Taking into account the initial conditions of the fundamental solutions and the Wronskian identity (\ref{wronski}), one sees that  $y_2(N) T_1 - y_1(N) T_2$ vanishes. Hence we have the formula
\begin{equation} \label{grad1}
\langle \n_{b,a} y_1, (v,u) \rangle (N) = -\frac{1}{a_{N}} \left( y_2(N) \langle y_1^\mbf2, (v,u) \rangle - y_1(N) \langle y_1 \cds y_2, (v,u) \rangle \right).
\end{equation}
Similarly, evaluating formula (\ref{diffint}) for $i=2$ and $m=N+1$ leads to
\begin{equation} \label{grad2}
\langle \n_{b,a} y_2, (v,u) \rangle (N\!+\!1) \! = \! -\frac{1}{a_{N}} \left( y_2(N\!+\!1) \langle y_1 \cds \, y_2, (v,u) \rangle \! - \! y_1(N\!+\!1) \langle y_2^\mbf2, (v,u) \rangle \right).
\end{equation}
Here we used that the value of the right side of (\ref{diffint}) does not change when we omit the term for $k = m = N+1$ in both sums.

It remains to prove the two formulas (\ref{nbdeltashift}) and (\ref{nadeltashift}). We first note that
\setlength\arraycolsep{2pt} {\begin{eqnarray}
  y_2(n, \l, S^m b, S^m a) & = & \frac{a_m}{a_N} \Big( y_2(n+m,\l,b,a) y_1(m,\l,b,a) \nonumber \\
&& \quad - y_1(n+m,\l,b,a) y_2(m,\l,b,a) \Big), \label{repry2shift}
\end{eqnarray}
since both sides of (\ref{repry2shift}) are solutions of $R_{S^m b, S^m a} y = \l y$ (for fixed $m \in \Z$) with the same initial conditions at $n = 0,1$. For $n=1$ this follows from the Wronskian identity (\ref{wronski}). Similarly, one shows that}
\begin{eqnarray}
  y_1(N+m,\l) & = & y_1(N,\l) y_1(m,\l) + y_1(N+1,\l) y_2(m,\l), \label{y1Nm} \\
  y_2(N+m,\l) & = & y_2(N,\l) y_1(m,\l) + y_2(N+1,\l) y_2(m,\l). \label{y2Nm}
\end{eqnarray}
for any $(b,a) \in \M$. Hence, suppressing the variable $\l$, we get
\setlength\arraycolsep{1pt} {\begin{eqnarray*}
  y_2(N,S^m b,S^m a) 
& = & \frac{a_m}{a_N} \Big( \big( y_2(N) y_1(m) + y_2(N+1) y_2(m) \big) \, y_1(m) \\
&& \quad - \big( y_1(N) y_1(m) + y_1(N+1) y_2(m) \big) \, y_2(m) \Big) \\
& = & \frac{a_m}{a_N} \Big( y_2(N) y_1(m)^2 + \big( y_2(N+1) - y_1(N) \big) y_1(m) y_2(m) \\
&& \qquad - y_1(N+1) y_2(m)^2 \Big).
\end{eqnarray*}
By (\ref{nbdelta1}) this leads to}
\begin{displaymath}
 y_2(N,S^m b,S^m a) = -a_m \frac{\partial \Delta_\l}{\partial b_m}
\end{displaymath}
and formula (\ref{nbdeltashift}) is established. To prove (\ref{nadeltashift}), we first conclude from (\ref{repry2shift}) that
\begin{eqnarray}
    \frac{a_N}{a_{m+1}} y_2(N-1, S^{m+1} b, S^{m+1} a) & = & y_2(N+m,b,a) y_1(m+1,b,a) \nonumber\\
&& \; - y_1(N+m,b,a) y_2(m+1,b,a) \label{y2n-1m+1exp}
\end{eqnarray}
and
\begin{eqnarray}
   \frac{a_N}{a_m} y_2(N+1, S^m b, S^m a) & = & y_2(N+m+1,b,a) y_1(m,b,a) \nonumber\\
&& \; - y_1(N+m+1,b,a) y_2(m,b,a). \label{y2n+1mexp}
\end{eqnarray}
Now expand the right hand sides of (\ref{y2n-1m+1exp}) and (\ref{y2n+1mexp}) according to (\ref{y1Nm}) and (\ref{y2Nm}). By (\ref{nadelta1}), the sum of (\ref{y2n-1m+1exp}) and (\ref{y2n+1mexp}) is $-a_N \frac{\partial \Delta_\l}{\partial a_m}$, thus proving (\ref{nadeltashift}).
\end{proof}

As a next step, we compute the gradients of the Dirichlet and periodic eigenvalues. In the following lemma, we consider the fundamental solution $y_1(\cdot, \mu)$ as an $N$-vector $y_1(j, \mu)_{1 \leq j \leq N}$. Let $\|y_1(\mu)\|^2 = \sum_{j=1}^N y_1(j, \mu)^2$, and denote by $\, \dot{} \,$ the derivative with respect to $\l$.

\begin{lemma} \label{dirsimple}
If $\mu$ is a Dirichlet eigenvalue of $L(b,a)$, then
\begin{equation} \label{dotynorm}
a_N y_1(N,\mu) \dot{y}_1(N+1,\mu) = \|y_1(\mu)\|^2 > 0.
\end{equation}
In particular, $\dot{y}_1(N+1,\mu) \neq 0$, which implies that all Dirichlet eigenvalues are simple.
\end{lemma}

\begin{proof}
This follows from adding up the relations (\ref{wronski2}).
\end{proof}

As the Dirichlet eigenvalues $(\mu_n)_{1 \leq n \leq N-1}$ of $L(b,a)$ coincide with the roots of $y_1(N+1, \mu)$ and these roots are simple, they are real analytic on $\M$. Similarly, the eigenvalues $\l_1$ and $\l_{2N}$ are real analytic on $\M$, whereas for any $1 \leq n \leq N-1$, $\l_{2n}$ and $\l_{2n+1}$ are real analytic on $\M \setminus D_n$. Note that for $(b,a) \in \M \setminus D_n$ and $i \in \{ 2n, 2n+1 \}$, we have $\dot{\Delta}_{\l_i} \neq 0$ as $\l_i$ is a simple eigenvalue.

\begin{prop} \label{muprod}
For any $1 \leq n \leq N-1$, the gradients of the periodic eigenvalues $\l_i$ ($i = 2n, 2n+1$) on $\M \setminus D_n$ and of the Dirichlet eigenvalues $\mu_n$ on $\M$ are given by
\begin{equation} \label{dlambda}
\n_{b,a} \l_i = -\frac{\n_{b,a} \Delta_\l|_{\l = \l_i}}{\dot{\Delta}_{\l_i}} = f_i^\mbf2 \quad \textrm{and} \quad \n_{b,a} \mu_n = g_n^\mbf2, 
\end{equation}
where we denote by $f_i$ the eigenvector of $L(b, a)$ associated to $\l_i$, normalized by
\begin{displaymath}
  \sum_{j=1}^N f_i(j)^2 = 1 \quad \textrm{and} \quad \big( f_i(1), f_i(2) \big) \in (\R_{>0} \times \R) \cup (\{ 0 \} \times \R_{>0}),
\end{displaymath}
and where $g_n = (g_n(j))_{1 \leq j \leq N}$ is the fundamental solution $y_1(\cdot, \mu_n)$ normalized so that $\sum_{j=1}^N g_n(j)^2 = 1$.
\end{prop}

\begin{proof}
We first show the second formula in (\ref{dlambda}). Differentiating $y_1(N+1, \mu_n) = 0$ with respect to $(b,a)$, one obtains
\begin{equation} \label{dmun1}
\n_{b,a} \mu_n = -\frac{\n_{b,a} y_1(N+1, \l)|_{\l = \mu_n}}{\dot{y}_1(N+1, \mu_n)}.
\end{equation}
Here we used that $\dot{y}_1(N+1, \mu_n) \neq 0$ by Lemma \ref{dirsimple}. To compute the gradient $\n_{b,a} y_1(N+1, \l)|_{\l = \mu_n}$, we evaluate (\ref{diffint}) for $i=1$ and $m = N+1$. In view of $y_1(N+1, \mu_n) = 0$ and taking into account (\ref{wrmu}), one then gets 
\begin{equation} \label{wrrepr1}
\n_{b,a} \mu_n = \frac{y_1^\mbf2(\mu_n)}{a_N y_1(N,\mu_n) \dot{y}_1(N+1,\mu_n)}.
\end{equation}
The claimed formula $\n_{b,a} \mu_n = g_n^\mbf2$ now follows from Lemma \ref{dirsimple}. By differentiating $\Delta_{\l_i} = \pm 2$ with respect to $(b,a)$, one obtains $\n_{b,a} \l_i = -\n_{b,a} \Delta_\l|_{\l = \l_i} / \dot{\Delta}_{\l_i}$ in a similar fashion. To see that $\n_{b,a} \l_i = f_i^\mbf2$, differentiate $R_{b,a} f_i = \l_i f_i$ with respect to $(b,a)$ in the direction $(v,u) \in \R^{2N}$,
\begin{displaymath}
  R_{b,a} \langle \n_{b,a} f_i, (v,u) \rangle (k) + (R_{v,u} f_i)(k) \!=\! \langle \n_{b,a} \l_i, (v,u) \rangle f_i(k) + \l_i \langle \n_{b,a} f_i, (v,u) \rangle (k),
\end{displaymath}
where $\langle \cdot, \cdot \rangle$ denotes the standard scalar product in $\R^{2N}$. Take the scalar product (in $\R^N$) of the above equation with $f_i$. Now use that
\begin{displaymath}
\langle \n_{b,a} f_i(v,u), R_{b,a} f_i \rangle_{\R^N} = \l_i \langle \n_{b,a} f_i(v,u), f_i \rangle_{\R^N},
\end{displaymath}
$\langle f_i, f_i \rangle_{\R^N} = 1$, and
\begin{displaymath}
  \langle R_{v,u} f_i, f_i \rangle_{\R^N} = \langle f_i^\mbf2, (v,u) \rangle_{\R^{2N}},
\end{displaymath}
to conclude that $\nba \l_i = f_i^\mbf2$ holds.
\end{proof}

To compute the Poisson brackets involving angle variables we need to establish some additional auxiliary results.
Recall from section \ref{coord} that for $1 \leq k,n \leq N-1$ with $k \neq n$ and $(b,a) \in \M$, $\beta_k^n$ is given by
\begin{equation} \label{reviewdef}
\beta_k^n = \int_{\l_{2k}}^{\mu_k^*} \frac{\psi_n(\l)}{\sqrt{\Delta^2_\l-4}} \; d\l, 
\end{equation}
whereas
\begin{equation} \label{reviewdef2}
\beta_n^n := \eta_n = \int_{\l_{2n}}^{\mu_n^*} \frac{\psi_n(\l)}{\sqrt{\Delta^2_\l-4}} \; d\l \;\; (\textrm{mod } 2\pi).
\end{equation}
By Theorem \ref{angleanalytic}, the functions $\b_k^n$ with $k \neq n$ are real analytic on $\M$, whereas $\b_n^n$, when considered mod $\pi$, is real analytic on $\M \setminus D_n$.

\begin{prop} \label{dbabkn}
Let $1 \leq k \leq N-1$ and $(b,a) \in \M$. If $\gamma_k > 0$ and $\l_{2k} = \mu_k$, then for any $1 \leq n \leq N-1$,
\begin{displaymath}
\n_{b,a} \beta_k^n = -\frac{\psi_n(\mu_k)}{a_N \dot{\Delta}_{\mu_k}} \, g_k \cds h_k, 
\end{displaymath}
where $h_k$ denotes the solution of $R_{b,a} y = \mu_k y$ orthogonal to $g_k$, i.e.
\begin{displaymath}
\sum_{j=1}^N g_k(j) h_k(j) = 0,
\end{displaymath}
satisfying the normalization condition $W(h_k, g_k)(N) = 1$.
\end{prop}

\begin{proof}
We use a limiting procedure first introduced in \cite{mcva1} for the nonlinear Schr\"odinger equation and subsequently used for the KdV equation in \cite{kama}, \cite{kapo}. We approximate $(b,a) \in \M$ with $\l_{2k}(b,a) = \mu_k(b,a) < \l_{2k+1}(b,a)$ by $(b',a') \in \textrm{Iso}(b,a)$ satisfying $\l_{2k}(b,a) < \mu_k(b',a') < \l_{2k+1}(b,a)$. For such $(b',a')$, using the substitution $\l = \l_{2k}+z$ in the integral of (\ref{reviewdef}), we obtain
\begin{equation} \label{anglesubst}
\beta_k^n(b',a') = \int_{\l_{2k}}^{\mu_k^*} \frac{\psi_n(\l)}{\sqrt{\Delta^2_\l-4}} \; d\l = \int_0^{\mu_k-\l_{2k}} \frac{\psi_n(\l_{2k}+z)} {\sqrt{z}\sqrt{D(z)}} dz,
\end{equation}
where $D(z) \equiv D(\l_{2k},z) := (\Delta^2(\l_{2k}+z)-4)/z$. Taking the gradient, undoing the substitution, and recalling the definition (\ref{munstarred}) of the starred square root then leads to
\begin{equation} \label{gradbknappr}
\nba \beta_k^n = \frac{\psi_n(\mu_k)}{\sqrt[*]{\Delta^2_{\mu_k} - 4}} (\nba \mu_k - \nba\l_{2k}) + E(b',a'),
\end{equation}
with the remainder term $E(b',a')$ given by
\begin{displaymath}
E(b',a') = \int_0^{\mu_k-\l_{2k}} \!\! \n_{b',a'} \left( \frac{\psi_n(\l_{2k}+z)}{\sqrt{D(\l_{2k},z)}} \right) \frac{dz}{\sqrt{z}}.
\end{displaymath}
As the gradient in the latter integral is a bounded function in $z$ near $z=0$, locally uniformly in $(b',a')$, it follows by the dominated convergence theorem that $\lim_{(b',a') \to (b,a)} E(b',a') = 0$.

The gradient $\nba \beta_k^n$ depends continuously on $(b,a) \in \M$, hence we can conclude by (\ref{gradbknappr}) that it can be written as
\begin{equation} \label{gradbkn}
\nba \beta_k^n = \lim_{(b',a') \to (b,a)} \frac{\psi_n(\mu_k)}{\sqrt[*]{\Delta^2_{\mu_k} - 4}} \left( \n_{b',a'} \mu_k - \n_{b',a'} \l_{2k} \right).
\end{equation}
The gradient of both sides of the Wronskian identity (\ref{wronski}),
$$ y_1(N, \l) y_2(N+1, \l) - y_1(N+1, \l) y_2(N, \l) = 1, $$
leads to
\begin{eqnarray}
y_1(N+1) \nba y_2(N) \, & + & \, y_2(N) \nba y_1(N+1) \label{grady1} \\
& = & y_2(N\!+\!1)\nba \Delta + (y_1(N)\!-\!y_2(N+1)) \nba y_2(N\!+\!1), \nonumber
\end{eqnarray}
The $\l$-derivative can then be computed to be
\begin{eqnarray}
y_1(N+1) \dot{y}_2(N) \, & + & \, \dot{y}_1(N+1)y_2(N) \nonumber \\
& = & y_2(N+1)\dot{\Delta} + (y_1(N)-y_2(N+1)) \dot{y}_2(N+1). \label{doty1}
\end{eqnarray}
Using (\ref{grady1}), (\ref{doty1}), and $y_1(N+1, \mu_k) = 0$, formula (\ref{dmun1}) for $\nba \mu_k$ leads to
\begin{equation} \label{nbamuk}
\nba \mu_k = -\frac{y_2(N+1) \nba \Delta + (y_1(N)-y_2(N+1))\nba y_2(N+1)} {y_2(N+1) \dot{\Delta} + (y_1(N)-y_2(N+1))\dot{y}_2(N+1)} \Big|_{\mu_k}.
\end{equation}
Further by (\ref{dlambda}),
\begin{equation} \label{nbal2k}
  \nba \l_{2k} = -\frac{\nba \Delta}{\dot{\Delta}}\Big|_{\l = \l_{2k}}.
\end{equation}
Now substitute (\ref{nbamuk}) and (\ref{nbal2k}) into $(\n_{b',a'} \mu_k - \n_{b',a'} \l_{2k})$ and use that by (\ref{munstarred}),
\begin{displaymath}
 \sqrt[*]{\Delta^2_{\mu_k} - 4} = (y_1(N) - y_2(N+1))|_{\mu_k}.
\end{displaymath}
We claim that \setlength\arraycolsep{0.1pt} {
\begin{equation}
\lim_{(b',a') \atop \to (b,a)} \!\! \frac{\nba \mu_k \!\! - \!\! \nba \l_{2k}} {\sqrt[*]{\Delta^2_{\mu_k} - 4}} = \frac{\dot{y}_2(N\!\!+\!\!1) \nba y_1(N) \! - \! \dot{y}_1(N) \nba y_2(N\!\!+\!\!1)}{\dot{\Delta} \, \dot{y}_1(N+1) y_2(N)} \Big|_{\l_{2k}}. \label{limcompl2}
\end{equation}}
Indeed, to obtain (\ref{limcompl2}) after the above mentioned substitutions, we split the fraction $\frac{1}{\sqrt[*]{\Delta^2_{\mu_k} - 4}}(\n_{b',a'} \mu_k - \n_{b',a'} \l_{2k})$ into two parts which are treated separately. In the first part we collect all terms in $\frac{1}{\sqrt[*]{\Delta^2_{\mu_k} - 4}}(\n_{b',a'} \mu_k - \n_{b',a'} \l_{2k})$ which contain $(y_1(N) - y_2(N+1))|_{\mu_k}$ in the nominator,
\begin{displaymath}
  I(a',b') := \frac{-\dot{\Delta}|_{\l_{2k}} \cdot \n_{b',a'} y_2(N+1)|_{\mu_k} + \n_{b',a'} \Delta|_{\l_{2k}} \cdot \dot{y}_2(N+1)|_{\mu_k}}{\dot{\Delta}|_{\l_{2k}} \cdot (y_2(N+1) \dot{\Delta} + (y_1(N) - y_2(N+1)) \dot{y}_2(N+1))|_{\mu_k}}.
\end{displaymath}
Using again (\ref{doty1}) we then get
\begin{displaymath}
  \lim_{(b',a') \to (b,a)} I(b',a') = \frac{\dot{y}_2(N+1) \nba y_1(N) - \dot{y}_1(N) \nba y_2(N+1)}{\dot{\Delta} \dot{y}_1(N+1) y_2(N)} \Big|_{\l_{2k}}.
\end{displaymath}
The second term is then given by
\begin{displaymath}
  II(b',a') = \frac{y_2(N+1)|_{\mu_k} \cdot (\n_{b',a'} \Delta|_{\l_{2k}} \cdot \dot{\Delta}|_{\mu_k} - \dot{\Delta}|_{\l_{2k}} \cdot \n_{b',a'} \Delta|_{\mu_k})}{\dot{\Delta}|_{\l_{2k}} \cdot (y_1(N) - y_2(N+1))|_{\mu_k} \cdot (\dot{y}_1(N+1) y_2(N))|_{\mu_k}}.
\end{displaymath}
Note that the nominator of $II(b',a')$ is of the order $O(\mu_k - \l_{2k})$. In view of (\ref{munstarred}), we have
\begin{displaymath}
  (y_1(N) - y_2(N+1))|_{\mu_k} = O(\sqrt{\mu_k - \l_{2k}})
\end{displaymath}
whereas the other terms in the denominator of $II(b',a')$ are bounded away from zero. Indeed, $\l_{2k}$ being a simple eigenvalue for $(b,a)$ means $\dot{\Delta}|_{\l_{2k}} \neq 0$ for $(b',a')$ near $(b,a)$. Further, use a version of (\ref{doty1}) in the case $\l_{2k} = \mu_k$ to conclude that
\begin{displaymath}
  \dot{y}_1(N+1) y_2(N) = y_2(N+1) \dot{\Delta}_{\l_{2k}}.
\end{displaymath}
Hence $\dot{y}_1(N+1) y_2(N)|_{\mu_k} \neq 0$ for $(b',a')$ near $(b,a)$ and $II(b',a')$ vanishes in the limit of $\mu_k \to \l_{2k}$.

Substituting (\ref{grad1}) and (\ref{grad2}) into (\ref{limcompl2}), we obtain
\begin{displaymath}
\frac{\dot{y}_2(N+1) \nba \, y_1(N) - \dot{y}_1(N) \nba \, y_2(N+1)} {\dot{\Delta} \, \dot{y}_1(N+1) y_2(N)} \Big|_{\mu_k} = -\frac{1}{a_N \dot{\Delta}} \, y_1 \cds y_0 
\end{displaymath}
with $y_0 = \frac{\dot{y}_2(N+1)}{\dot{y}_1(N+1)} y_1 - y_2$. Hence
\begin{displaymath}
\nba \beta_k^n = -\frac{\psi_n(\mu_k)}{a_N \dot{\Delta}(\mu_k)} \, y_1 \cds y_0. 
\end{displaymath}
Since $\beta_k^n$ is invariant under the  translation $b \mapsto b + t (1, \ldots, 1)$, the scalar product $\langle \nba \beta_k^n, (\mathbf{1}, \mathbf{0}) \rangle_{\R^{2N}}$ vanishes. Hence
$$ 0 = \sum_{j=1}^N \frac{\partial \beta_k^n}{\partial b_j} = -\frac{\psi_n(\mu_k)}{a_N \dot{\Delta}(\mu_k)} \sum_{j=1}^N y_1(j) y_0(j). $$
It means that $y_1$ and $y_0$ are orthogonal to each other. Finally we introduce $h_k := \| y_1 \| y_0$ and verify that
\setlength\arraycolsep{2pt} {
\begin{displaymath}
W(h_k, g_k) = W \left( \| y_1 \| y_0, \frac{y_1}{\| y_1 \|} \right) = W(y_0, y_1) = W \left( \frac{\dot{y}_2(N+1)}{\dot{y}_1(N+1)} y_1 - y_2, y_1 \right).
\end{displaymath}}
By (\ref{wronski}), it then follows that
\begin{displaymath}
  W(h_k, g_k) = -W(y_2, y_1) = W(y_1,y_2).
\end{displaymath}
Hence by (\ref{wronskispecial}),
\begin{displaymath}
  W(h_k, g_k)(N) = W(y_1, y_2)(N) = 1.
\end{displaymath}
This completes the proof of Proposition \ref{dbabkn}.
\end{proof}

\section{Orthogonality relations} \label{orthrelsection}

In Propositions \ref{discrprod}, \ref{muprod}, and \ref{dbabkn}, we have
expressed the gradients of $\Delta_\l$, $\mu_n$, and, on a subset of $\M$, of $\beta_k^n$ in terms of products of fundamental solutions of the difference equation (\ref{diff}). In this section we establish orthogonality relations between such products - see \cite{bggk} for similar computations. Recall that in (\ref{vector2n}) we have introduced for arbitrary sequences $(v_j)_{j \in \Z}$, $(w_j)_{j \in \Z}$ the $2N$-vector $v \cds w$.

\begin{lemma} \label{orthfund}
For any $(b,a) \in \M$, let $v_1$, $w_1$ and $v_2$, $w_2$ be pairs of solutions of (\ref{diff}) for arbitrarily given real numbers $\mu$ and $\l$, respectively. Then
\begin{equation} \label{orthfundformula}
\frac{2(\l-\mu)}{a_1 a_N} \langle v_1 \cds w_1, J (v_2 \cds w_2) \rangle = V + B,
\end{equation}
where
\begin{equation} \label{wdef}
V := \big( \Wa \cdot S \, \Wb \big) \big|^N_0 + \big( S \, \Wa \cdot \Wb \big) \big|^N_0
\end{equation}
with $\Wa$ and $\Wb$ denoting the Wronskians $\Wa := W(v_1, w_2)$, $\Wb := W(w_1, v_2)$, and where $B$ is given by
\begin{equation}  \label{bdef}
B \!\! := \!\! \frac{(\l\!-\!\mu)}{a_1} \Big( (v_1 \cdot w_1)|_1^{N+1} (v_2 \cds w_2)(2N) \!-\! (v_2 \cdot w_2)|_1^{N+1} (v_1 \cds w_1)(2N) \Big).
\end{equation}
\end{lemma}

\begin{proof}
We prove (\ref{orthfundformula}) by a straightforward calculation, using the recurrence property (\ref{wronski1}) of the Wronskian sequences $W_1$ and $W_2$. By the definition (\ref{jdef}) of $J$ we can write
$$ 2 \, \langle v_1 \cds w_1, J (v_2 \cds w_2) \rangle = E_1 + B_1, $$
where
\begin{eqnarray*}
E_1 & := & \sum_{k=1}^N a_k \big[ (v_1 \cds w_1)(k) (v_2 \cds w_2)(N+k) - (v_1 \cds w_1)(k+1) (v_2 \cds w_2)(N+k) \\
&& \qquad - (v_1 \cds w_1)(N+k) (v_2 \cds w_2)(k) + (v_1 \cds w_1)(N+k) (v_2 \cds w_2)(k+1) \big]
\end{eqnarray*}
and
\setlength\arraycolsep{1pt}{ \begin{eqnarray*}
B_1 & := & a_N \, \big( (v_1 \cds w_1)(N+1) - (v_1 \cds w_1)(1) \big) \, (v_2 \cds w_2)(2N) \\
&& + a_N \, \big( (v_2 \cds w_2)(1) - (v_2 \cds w_2)(N+1) \big) \, (v_1 \cds w_1)(2N).
\end{eqnarray*}}
Let us first consider $E_1$. Calculating the products $v_j \cds w_j$ according to (\ref{vector2n}), we obtain, after regrouping,
\begin{eqnarray*}
E_1 & 
 = & \sum_{k=1}^{N} a_k \big[ (v_2(k) w_1(k) + v_2(k+1) w_1(k+1)) \Wa(k) \\
&& \,\, + (v_1(k) w_2(k) + v_1(k+1) w_2(k+1)) \Wb(k) \big] \\
&& \!\! + B_2
\end{eqnarray*}
with
\begin{eqnarray*}
B_2 & := & a_N \big( v_1(N+1) w_1(N+1) - (v_1 \cds w_1)(N+1) \big) \, (v_2 \cds w_2)(2N) \\
&& +a_N \big( (v_2 \cds w_2)(N\!+\!1) - v_2(N\!+\!1) w_2(N\!+\!1) \big) \, (v_1 \cds w_1)(2N).
\end{eqnarray*}
Multiply $E_1$ by $(\l-\mu)$ and use the recurrence relation (\ref{wronski1}) to express $(\l-\mu) v_2(k) w_1(k)$, $(\l-\mu) v_2(k+1) w_1(k+1)$, $(\l-\mu) v_1(k) w_2(k)$, and $(\l-\mu) v_1(k+1) w_2(k+1)$ in terms of the Wronskians $\Wa$ and $\Wb$ to get
\setlength\arraycolsep{1pt}
 {\begin{eqnarray*}
(\l-\mu) E_1 & = & \sum_{k=1}^{N} [a_k a_{k+1} (\Wa(k) \Wb(k+1) + \Wa(k+1) \Wb(k)) \\
&& - a_{k-1} a_k (\Wa(k-1) \Wb(k) + \Wa(k) \Wb(k-1))] \\
&& + (\l-\mu) B_2.
\end{eqnarray*}}
The sum on the right hand side of the latter identity is a telescoping sum and equals the term $a_1 a_N V$ with $V$ defined in (\ref{wdef}). In a straightforward way one sees that $\frac{(\l-\mu)}{a_1 a_N} (B_1 + B_2)$ equals the expression $B$ defined by (\ref{bdef}), hence formula (\ref{orthfundformula}) is established.
\end{proof}

\begin{cor} \label{dadb}
For any $\l, \mu \in \C$,
\begin{equation} \label{dadb0}
\{ \Delta_\l, \Delta_\mu \}_J = 0.
\end{equation}
\end{cor}

\begin{proof}
By the formula (\ref{dpr}) for the gradient of $\Delta_{\l}$,
\begin{displaymath}
  \{ \Delta_\l, \Delta_\mu \}_J = \langle \nba \Delta_\l, J \nba \Delta_{\mu} \rangle
\end{displaymath}
is a linear combination of terms of the form $\langle v_1 \cds w_1, J (v_2 \cds w_2) \rangle$ for pairs of fundamental solutions $v_1$, $w_1$ and $v_2$, $w_2$ of (\ref{diff}) for $\mu$ and $\l$, respectively. In view of (\ref{nbdeltashift}) and (\ref{nadeltashift}), $\n_b \Delta_\l$ and $\n_a \Delta_\l$ are both $N$-periodic. In the case $\l \neq \mu$ we use Lemma \ref{orthfund} and note that the boundary terms (\ref{wdef}) and (\ref{bdef}) in Lemma \ref{orthfund} vanish, hence $\{ \Delta_\l, \Delta_\mu \}_J = 0$. In the case $\l = \mu$ the identity (\ref{dadb0}) follows from the skew-symmetry of $\{ \cdot, \cdot \}_J$.
\end{proof}

\begin{cor} \label{lkdeltal}
Let $1 \leq k  \leq 2N$ and $\l \in \C$. On the open subset of $\M$ where $\l_k$ is a simple eigenvalue of $Q(b,a)$ one has
\begin{displaymath}
  \{ \l_k, \Delta_\l \}_J = 0.
\end{displaymath}
\end{cor}

\begin{proof}
Using formula (\ref{dlambda}) for $\nba \l_k$, we conclude from Corollary \ref{dadb} that
\begin{displaymath}
  \{ \l_k, \Delta_\l \}_J = -\frac{1}{\dot{\Delta}_{\l_k}} \{ \Delta_\mu, \Delta_\l \}_J|_{\mu = \l_k} = 0.
\end{displaymath}
\end{proof}

\begin{cor} \label{dirlemma}
Let $\mu_n$ be the $n$-th Dirichlet eigenvalue of $L(b,a)$ and $\l
\neq \mu_n$ a real number. Then \setlength\arraycolsep{2pt} {
\begin{eqnarray}
(\l - \mu_n) \langle y_1^\mbf2(\mu_n), J y_1^\mbf2(\l) \rangle & = & \left( a_N \frac{y_1(N+1,\l)}{y_2(N+1,\mu_n)} \right)^2 \label{1j1} \\
(\l - \mu_n) \langle y_1^\mbf2(\mu_n), J y_1(\l) \cds y_2(\l) \rangle & = & a_N^2 \frac{y_1(N+1,\l) y_2(N+1,\l)}{y_2(N+1,\mu_n)^2} \label{1j12} \\
(\l - \mu_n) \langle y_1^\mbf2(\mu_n), J y_2^\mbf2(\l) \rangle & = & a_N^2 \left( \left( \frac{y_2(N+1,\l)}{y_2(N+1,\mu_n)} \right)^2 - 1 \right) \label{1j2}
\end{eqnarray}}
\end{cor}

\begin{proof}
The three stated identities follow from Lemma \ref{orthfund}, using that $y_1(N+1,\mu_n) = 0$, $y_1(2, \mu_n) = -a_N/a_1$, and, by the Wronskian identity (\ref{wrmu}), $y_1(N, \mu_n) \cdot \\ y_2(N+1, \mu_n) = 1$.
\end{proof}

\begin{cor} \label{mndl}
Let $\mu_n$ be the $n$-th Dirichlet eigenvalue of $L(b,a)$ and $\l \neq \mu_n$ a real number. Then
\begin{equation} \label{mneq1}
\{ \mu_n, \Delta_\l \}_J = \frac{y_1(N+1,\l)}{\dot{y}_1(N+1,\mu_n)} \frac{\sqrt[*]{\Delta^2_{\mu_n}-4}}{\l-\mu_n}.
\end{equation}
\end{cor}

\begin{proof}
By (\ref{wrrepr1}), combined with (\ref{wrmu}), we get
\begin{equation} \label{mndlrepr}
\{ \mu_n, \Delta_\l \}_J = \frac{y_2(N+1,\mu_n)}{a_N \dot{y}_1(N+1,\mu_n)}\langle y_1^\mbf2(\mu_n), J \nba \Delta_\l \rangle.
\end{equation}
Substituting the formula (\ref{dpr}) for $J \nba \Delta_\l$ we obtain
\begin{eqnarray}
\langle y_1^\mbf2(\mu_n), J \nba \Delta_\l \rangle & = & -\frac{1}{a_N} y_2(N,\l) \langle y_1^\mbf2(\mu_n), J y_1^\mbf2(\l) \rangle \nonumber \\
&& -\frac{1}{a_N} (y_2(N+1,\l) - y_1(N,\l)) \langle y_1^\mbf2(\mu_n), J y_1(\l) \cds y_2(\l) \rangle \nonumber \\
&& +\frac{1}{a_N} y_1(N+1,\l) \langle y_1^\mbf2(\mu_n), J y_2^\mbf2(\l) \rangle. \label{y1munjd}
\end{eqnarray}

To evaluate the right side of (\ref{y1munjd}), we apply Corollary
\ref{dirlemma} and get \setlength\arraycolsep{0pt} {
\begin{eqnarray*}
\frac{\l\!\!-\!\!\mu_n}{a_1 a_N} \langle y_1^\mbf2(\mu_n), J \nba \Delta_\l \rangle & = & \frac{1}{a_1 y_2(N+1,\mu_n)^2} \Big( -y_2(N,\l) y_1(N+1,\l)^2 \\
&& \!\!\! - (y_2(N+1,\l) - y_1(N,\l)) y_1(N+1,\l) y_2(N+1,\l) \\
&& + y_1(N+1,\l) y_2(N+1,\l)^2 \Big) - \frac{y_1(N+1,\l)}{a_1}.
\end{eqnarray*}}
Using the Wronskian identity (\ref{wronskispecial}), the sum of the terms in the square bracket of the latter expression simplifies, and one obtains
\setlength\arraycolsep{2pt} {
\begin{eqnarray}
\frac{\l-\mu_n}{a_1 a_N} \langle y_1^\mbf2(\mu_n), J \nba \Delta_\l \rangle & = & \frac{y_1(N+1,\l)}{a_1 y_2(N+1,\mu_n)^2} - \frac{y_1(N+1,\l)}{a_1} \nonumber \\
& = & \frac{y_1(N+1,\l)}{a_1} (y_1(N,\mu_n)^2 - 1), \label{auxsquare}
\end{eqnarray}
where for the latter equality we again used (\ref{wronskispecial}). Substituting (\ref{auxsquare}) into (\ref{mndlrepr}), we get
\begin{eqnarray*}
\frac{\l-\mu_n}{a_1 a_N} \{ \mu_n, \Delta_\l \}_J & = & \frac{y_2(N+1,\mu_n) y_1(N+1,\l)}{a_1 a_N \dot{y}_1(N+1,\mu_n)} (y_1(N,\mu_n)^2 - 1) \\
& 
= & \frac{y_1(N+1,\l)}{a_1 a_N \dot{y}_1(N+1,\mu_n)} \sqrt[*]{\Delta^2_{\mu_n}-4},
\end{eqnarray*}
where we used that, by the definition of the starred square root (\ref{munstarred}),
$$ \sqrt[*]{\Delta^2_{\mu_n}-4} = y_1(N, \mu_n) - y_2(N+1, \mu_n) = y_2(N+1, \mu_n) (y_1(N, \mu_n)^2 - 1). $$
This proves (\ref{mneq1}).}
\end{proof}

\begin{prop} \label{thndeltal}
For any $\l \in \R$, $1 \leq n \leq N-1$, and $(b,a) \in \M \setminus D_n$,
\begin{equation} \label{mneq2}
\{ \theta_n, \Delta_\l \}_J = \psi_n(\l).
\end{equation}
\end{prop}

\begin{proof}
Recall that $\theta_n = \sum_{k=1}^{N-1} \beta_k^n$ (mod $2\pi$) with $\beta_k^n$ given by (\ref{reviewdef})-(\ref{reviewdef2}). To compute $\{ \beta_k^n, \Delta_\l \}_J$, we first consider the case where $(b,a) \notin \bigcup_{k=1}^{N-1} D_k$ and $\l_{2k} < \mu_k < \l_{2k+1}$ for any $1 \leq k \leq N-1$. Then $\l_{2k}$ and $\mu_k^*$ are smooth near $(b,a)$ and, by Leibniz's rule, we get
\setlength\arraycolsep{2pt} {
\begin{eqnarray*}
  \{ \beta_k^n, \Delta_\l \}_J & = & \Bigg( \int_{\l_{2k}}^{\mu_k^*} \{ \frac{\psi_n(\mu)}{\sqrt{\Delta_\mu^2-4}}, \Delta_\l \}_J \; d\mu \\
&& + \frac{\psi_n(\mu_k)}{\sqrt[*]{\Delta_{\mu_k}^2-4}} \{ \mu_k, \Delta_\l \}_J - \frac{\psi_n(\l_{2k})}{\sqrt[*]{\Delta_{\l_{2k}}^2-4}} \{ \l_{2k}, \Delta_\l \}_J \Bigg).
\end{eqnarray*}}
By Corollary \ref{lkdeltal}, $\{ \l_{2k}, \Delta_\l \}_J = 0$. Moreover, as the gradient $\nba \frac{\psi_n(\mu)}{\sqrt{\Delta_{\mu}^2-4}}$ is orthogonal to $T_{b,a}\,\textrm{Iso}\,(b,a)$ and $J \nba \Delta_\l \in T_{b,a}\,\textrm{Iso}\,(b,a)$ it follows that the Poisson bracket $\{ \frac{\psi_n(\mu)}{\sqrt{\Delta_{\mu}^2-4}}, \Delta_\l \}_J$ vanishes for any $\mu$ in the isolating neighborhood $U_n$ of $G_n$. Hence
\begin{displaymath}
  \{ \beta_k^n, \Delta_\l \}_J = \frac{\psi_n(\mu_k)}{\sqrt[*]{\Delta_{\mu_k}^2-4}} \{ \mu_k, \Delta_\l \}_J.
\end{displaymath}
By (\ref{mneq1}), we then obtain
$$ \{ \theta_n, \Delta_\l \}_J = \sum_{k=1}^{N-1} \frac{\psi_n(\mu_k)}{\dot{y}_1(N+1,\mu_k)} \frac{y_1(N+1,\l)}{\l - \mu_k} = \psi_n(\l), $$
where for the latter equality we used that $\sum_{k=1}^{N-1} \frac{\psi_n(\mu_k)}{\dot{y}_1(N+1,\mu_k)} \frac{y_1(N+1,\l)}{\l - \mu_k}$ and $\psi_n(\l)$ are both polynomials in $\l$ of degree at most $N-2$ which agree at the $N-1$ points $(\mu_k)_{1 \leq k \leq N-1}$.

In the general case, where $(b,a) \in \M \setminus D_n$ and the Dirichlet eigenvalues are arbitrary, $\l_{2k} \leq \mu_k \leq \l_{2k+1}$ for any $1 \leq k \leq N-1$, the claimed result follows from the case treated above by continuity.
\end{proof}

\begin{prop} \label{bknblm}
Let $1 \leq n,m,k,l \leq N-1$ and let $(b,a) \in \M$ with $\l_{2i}(b,a) = \mu_i(b,a)$ for $i = k,l$. Then
$$ \{ \beta_k^n, \beta_l^m \}_J = 0. $$
\end{prop}

\begin{proof}
In view of Proposition \ref{dbabkn}, this amounts to showing that the scalar product $\langle (g_k \cds \, h_k), J (g_l \cds \, h_l) \rangle$ vanishes. For $k = l$, this follows from the skew-symmetry of the Poisson bracket, hence we can assume $k \neq l$. We apply Lemma \ref{orthfund} with $v_1 := g_k$, $w_1 := h_k$, $v_2 := h_l$ and $w_2 := g_l$, which implies that $\Wa = W(g_k, g_l)$ and $\Wb = W(h_k, h_l)$. Since $g_k(1)$, $g_l(1)$, $g_k(N+1)$ and $g_l(N+1)$ all vanish, we conclude that $W_1(N) = W_1(0) = 0$ and $(SW_1)(N) = (SW_1)(0) = 0$, hence the expressions $V$ and $E$, defined in (\ref{wdef}) and (\ref{bdef}), vanish. This proves the claim.
\end{proof}

\section{Canonical relations} \label{poissonchapter}

In this section we complete the proof of Theorem \ref{fundthm} and Corollary \ref{leafcor}. In particular we show that the variables $(I_n)_{1 \leq n \leq N-1}$, $(\theta_n)_{1 \leq n \leq N-1}$ satisfy the canonical relations stated in Theorem \ref{fundthm}.

Using the results of the preceding sections, we can now compute the Poisson brackets among the action and angle variables introduced in section \ref{coord}.

\begin{theorem} \label{orthrel}
The action-angle variables $(I_n)_{1 \leq n \leq N-1}$ and $(\theta_n)_{1 \leq n \leq N-1}$ satisfy the following canonical relations for $1 \leq n,m \leq N-1$:

\begin{itemize}

\item[(i)] on $\M$,
\begin{equation} \label{inim}
\{ I_n, I_m \}_J = 0;
\end{equation}

\item[(ii)] on $\M \setminus D_n$,
\begin{equation} \label{thnim}
\{ \theta_n, I_m \}_J = -\{ I_m, \theta_n \}_J = -\delta_{nm}.
\end{equation}
\end{itemize}
\end{theorem}

\begin{proof}
Recall that $\frac{d}{dt} \, \textrm{arcosh}\,(t) = (t^2-1)^{-\frac{1}{2}}$. Hence for any $(b,a) \in \M$
\begin{displaymath}
  I_n = \frac{1}{2\pi} \int_{\Gamma_n} \l \frac{d}{d\l}\,\textrm{arcosh}\,\left| \frac{\Delta_\l}{2} \right| \, d\l
\end{displaymath}
and therefore
\begin{displaymath}
  \nba I_n = \frac{1}{2\pi} \int_{\Gamma_n} \l \frac{d}{d\l} \frac{\nba \Delta_\l}{\sqrt[c]{\Delta_\l^2-4}} \, d\l.
\end{displaymath}
Integrating by parts we get
\begin{equation} \label{gradcontour}
\nba I_n = -\frac{1}{2\pi} \int_{\Gamma_n} \frac{\nba \Delta_\l} {\sqrt[c]{\Delta^2_\l-4}} d\l.
\end{equation}
As $\{ \Delta_{\l}, \Delta_{\mu} \}_J = 0$ for all $\l, \mu \in \C$ by Corollary \ref{dadb}, it follows that $\{ I_n, I_m \}_J = 0$  on $\M$ for any $1 \leq n,m \leq N-1$.

To prove (\ref{thnim}), use (\ref{gradcontour}) and then Proposition \ref{thndeltal} to get
$$ \{ \theta_n, I_m \}_J = -\frac{1}{2\pi} \int_{\Gamma_m} \frac{\{ \theta_n, \Delta_\l \}_J}{\sqrt[c]{\Delta^2_\l-4}} d\l = -\frac{1}{2\pi} \int_{\Gamma_m} \frac{\psi_n(\l)}{\sqrt[c]{\Delta^2_\l-4}} d\l = -\delta_{nm}, $$
by the normalizing condition (\ref{psi}) of $\psi_n$.
\end{proof}

To prove that the angles $(\theta_n)_{1 \leq n \leq N-1}$ pairwise Poisson commute we need the following lemma. We denote by $K = K(b,a)$ the index set of the open gaps, i.e.
$$ K(b,a) = \{ 1 \leq n \leq N-1: \gamma_n(b,a) > 0 \}. $$
\begin{lemma} \label{lemmaindep}
At every point $(b,a)$ in $\M$, the set of vectors
\begin{itemize}
\item[(i)] \qquad \qquad \qquad \qquad $\big( (\nba I_n)_{n \in K}, \nba C_1, \nba C_2 \big)$
\item[and]
\item[(ii)] \qquad \qquad \qquad \qquad \qquad \qquad $(J \, \nba I_n)_{n \in K}$
\end{itemize}
are both linearly independent.
\end{lemma}

\begin{proof}
The claimed statements follow from the orthogonality relations stated in Theorem \ref{orthrel}: Let $(b,a) \in \M$ and suppose that for some real coefficients $(r_n)_{n \in K} \subseteq \R$ and $s_1, s_2 \in \R$ we have
\begin{displaymath}
\sum_{n \in K}r_n \nba I_n + s_1 \nba C_1 + s_2 \n C_2 = 0.
\end{displaymath}
For any $m \in K$, take the scalar product of this identity with $J \nba \theta_m$. Using that $\{ I_n, \theta_m \}_J = \delta_{nm}$ and that $C_1$ and $C_2$ are Casimir functions of $\{ \cdot, \cdot \}_J$ one obtains
$$ 0 = \sum_{n \in K} r_n \{ I_n, \theta_m \}_J = \sum_{n \in K} r_n \delta_{nm} = r_m. $$
Thus $r_m = 0$ for all $m \in K$, and it follows that $s_1 \nba C_1 + s_2 \nba C_2 = 0$. By (\ref{c1grad}) and (\ref{c2grad}), $\nba C_1$ and $\nba C_2$ are linearly independent, hence $s_1 = s_2 = 0$. This shows (i). The proof of (i) also shows that (ii) holds.
\end{proof}

\begin{theorem} \label{orthrel2}
In addition to the canonical relations stated in Theorem \ref{orthrel}, the angle variables $(\theta_n)_{1 \leq n \leq N-1}$ satisfy for any $1 \leq n,m \leq N-1$ on $\M \setminus (D_n \cup D_m)$
\begin{equation} \label{thnthm}
\{ \theta_n, \theta_m \}_J = 0.
\end{equation}
\end{theorem}

\begin{proof}
Let $1 \leq n,m \leq N-1$. By continuity, it suffices to prove the identity (\ref{thnthm}) for $(b,a) \in \M \setminus \big( \bigcup_{l=1}^{N-1} D_l \big)$. Let $(b,a)$ be an arbitrary element in $\M \setminus \big( \bigcup_{l=1}^{N-1} D_l \big)$. Recall that Iso$(b,a)$ denotes the set of all elements $(b',a')$ in $\M$ with spec$(Q_{b',a'}) =$ spec$(Q_{b,a})$,
\begin{displaymath}
  \textrm{Iso}\,(b,a) = \{ (b',a') \in \M: \Delta(\cdot, b', a') = \Delta(\cdot, b, a) \}.
\end{displaymath}
Then Iso$(b,a)$ is a torus contained in $\M \setminus \big( \bigcup_{l=1}^{N-1} D_l \big)$, and as all eigenvalues of $Q(b,a)$ are simple, its dimension is $N-1$. By Lemma \ref{lemmaindep}, at any point $(b',a') \in$ Iso$(b,a)$, the vectors $(J \, \n_{b',a'} I_k)_{1 \leq k \leq N-1}$ are linearly independent. Using the formula (\ref{gradcontour}) for the gradient of $I_k$, one sees that, by Corollary \ref{dadb}, for any $\mu \in \R$, $1 \leq k \leq N-1$,
\begin{displaymath}
  \langle \n_{b',a'} \Delta_\mu, J \, \n_{b',a'} I_k  \rangle = -\frac{1}{2\pi} \int_{\Gamma_n} \frac{\{ \Delta_\mu, \Delta_\l \}_J}{\sqrt[c]{\Delta_\l^2-4}} \, d\l = 0.
\end{displaymath}
Hence for any $(b',a') \in$ Iso$(b,a)$,
\begin{displaymath}
  (J \, \n_{b',a'} I_k)_{1 \leq k \leq N-1} \in T_{b',a'} \, \textrm{Iso} \, (b,a),
\end{displaymath}
and therefore these vectors form a basis of $T_{b',a'} \, \textrm{Iso} \, (b,a)$.

To prove the identity (\ref{thnthm}), we apply the Jacobi identity
\begin{displaymath}
  \{ F, \{ G, H \}_J\}_J + \{ G, \{ H, F \}_J\}_J + \{ H, \{ F, G \}_J\}_J = 0
\end{displaymath}
to the functions $I_k$, $\theta_n$ and $\theta_m$. Since by Theorem \ref{orthrel}, $\{ I_k, \theta_n \}_J = \delta_{kn}$, we obtain
\begin{displaymath}
  \{ I_k, \{ \theta_n, \theta_m \}_J\}_J = 0 \;\; \textrm{on} \;\; \M \setminus \left( \bigcup_{l=1}^{N-1} D_l \right) \;\; \textrm{for any} \;\; 1 \leq k \leq N-1.
\end{displaymath}
It then follows by the above considerations that $\n_{b',a'} \{ \theta_n, \theta_m \}_J$ is orthogonal to $T_{b',a'} \, \textrm{Iso} \, (b,a)$ for all $(b',a') \in$ Iso$(b,a)$, i.e. $\{ \theta_n, \theta_m \}_J$ is constant on Iso$(b,a)$,
\begin{displaymath}
  \{ \theta_n, \theta_m \}_J (b',a') = \{ \theta_n, \theta_m \}_J (b,a) \quad \forall \, (b',a') \in \, \textrm{Iso} \, (b,a).
\end{displaymath}
By \cite{moer}, Theorem 2.1, there exists a unique element $(b',a') \in$ Iso$(b,a)$ satisfying $\mu_k(b',a') = \l_{2k}(b,a)$ for all $1 \leq k \leq N-1$. The claimed identity (\ref{thnthm}) then follows from Proposition \ref{bknblm}.
\end{proof}

\begin{proof}[Proof of Theorem \ref{fundthm}]
By Theorem \ref{analytic} and Theorem \ref{angleanalytic}, the action and angle variables introduced in Definitions \ref{actionsdef} and \ref{angledefinition}, respectively, have the claimed analyticity properties. The canonical relations among these variables have been verified in Theorem \ref{orthrel} and Theorem \ref{orthrel2}, and the relations $\{ C_i, I_n \}_J = 0$ (on $\M$) and $\{ C_i, \theta_n \}_J = 0$ (on $\M \setminus D_n$) follow from the fact that $C_1$ and $C_2$ are Casimir functions. It remains to show that the actions Poisson commute with the Toda Hamiltonian. To this end note that that the Hamiltonian $H$ can be written as
\begin{displaymath}
  H = \frac{1}{2} \sum_{n=1}^N b_n^2 + \sum_{n=1}^N a_n^2 = \frac{1}{2} \, \textrm{tr} \, (L(b,a)^2) = \frac{1}{2} \sum_{j=1}^N (\l_j^+)^2
\end{displaymath}
where $(\l_j^+)_{1 \leq j \leq N}$ are the $N$ eigenvalues of $L(b,a)$. Recall that on the dense open subset $\M \setminus \cup_{k=1}^{N-1} D_k$ of $\M$, the $\l_i^+$'s ($1 \leq i \leq N$) are simple eigenvalues and hence real analytic. It then follows by (\ref{gradcontour}) that for any $1 \leq n \leq N-1$, 
\begin{displaymath}
  \{ H, I_n \}_J = \sum_{i=1}^N \l_i^+ \{ \l_i^+, I_n \}_J = -\sum_{i=1}^N \frac{\l_i^+}{2\pi} \int_{\Gamma_n} \frac{\{ \l_i^+, \Delta_\l \}_J} {\sqrt[c]{\Delta^2_\l-4}} d\l = 0,
\end{displaymath}
where for the latter identity we used Corollary \ref{lkdeltal}. Hence for any $1 \leq n \leq N-1$,
\begin{displaymath}
  \{ H, I_n \}_J = 0 \quad \textrm{on} \; \M \setminus \cup_{k=1}^{N-1} D_k.
\end{displaymath}
By continuity it then follows that $\{ H, I_n \}_J = 0$ everywhere on $\M$.
\end{proof}

\begin{proof}[Proof of Corollary \ref{leafcor}]
Since for any $\beta \in \R$ and $\alpha > 0$ the symplectic leaf $\Mba$ is a submanifold of $\M$ of dimension $2(N-1)$, there are at most $N-1$ independent integrals in involution on $\Mba$. For any given $(b,a) \in \Mmba$ let $\pmba$ denote the orthogonal projection $T_{b,a}\M \to T_{b,a}\Mmba$. Then the gradient of the restriction ${I_n}|_{\Mmba}$ of $I_n$ to $\Mmba$ ($1 \leq n \leq N-1$) is given by $\pmba \nba I_n$. By Lemma \ref{lemmaindep} the vectors $(\pmba \nba I_n)_{n \in K}$ are linearly independent. As $\Mmba \setminus \cup_{k=1}^{N-1}D_k$ is dense in $\Mmba$, it then follows that $({I_n}|_{\Mmba})_{1 \leq n \leq N-1}$ are functionally independent. Finally, as $C_1$ and $C_2$ are Casimir functions of $\{ \cdot, \cdot \}_J$, it follows that for any $(b,a) \in \Mmba$
\setlength\arraycolsep{2pt} {
\begin{eqnarray*}
  \{ {I_n}|_{\Mmba}, {I_m}|_{\Mmba} \}_J (b,a) & = & \langle \pmba \nba I_n, \pmba J \nba I_m \rangle \\
& = & \langle \nba I_n, J \nba I_m \rangle \; = \; \{ I_n, I_m \}_J \; = \; 0,
\end{eqnarray*}}
i.e. the restrictions ${I_n}|_{\Mmba}$ of $I_n$ ($1 \leq n \leq N-1$) are in involution.
\end{proof}

\appendix

\section{Proof of Lemma \ref{est1}} \label{proofest}

In this Appendix, we prove Lemma \ref{est1}. It turns out that the proof in (\cite{bbgk}, p. 601-602) of the special case where the parameter $\a$ in (\ref{hamtodapq}) equals $1$ can be adapted for arbitrary values.

\begin{proof}[Proof of Lemma \ref{est1}]
Let $(b,a)$ be an arbitrary element in $\M$ and $1 \leq n \leq N-1$. First note that $I_n = \frac{1}{\pi} \int_{\l_{2n}}^{\l_{2n+1}}  \textrm{arcosh} |\frac{1}{2} \Delta(\l)| \, d\l$ and use $\frac{d}{dt}\,\textrm{arcosh}\,(t) = \frac{1}{\sqrt{t^2-1}}$ to obtain
$$ I_n = \frac{1}{\pi} \int_{\l_{2n}}^{\l_{2n+1}}\int_{1}^{|\Delta(\l)|/2} \frac{1}{\sqrt{t^2-1}} \; dt \; d\l. $$
Since the integrand of the inner integral is nonincreasing, we estimate it from below by its value at $\frac{|\Delta(\l)|}{2}$. This leads to
\begin{equation} \label{integral}
I_n \geq \frac{1}{\pi} \int_{\l_{2n}}^{\l_{2n+1}} \frac{\sqrt{|\Delta(\l)|-2}}{\sqrt{|\Delta(\l)|+2}} \; d\l.
\end{equation}
We will show below that for $\l_{2n} \leq \l \leq \l_{2n+1}$
\begin{equation} \label{discrest}
\frac{\sqrt{|\Delta(\l)|-2}}{\sqrt{|\Delta(\l)|+2}} \geq \frac{\sqrt{\l-\l_{2n}}\sqrt{\l_{2n+1}-\l}}{\l_{2N}-\l_1}.
\end{equation}
We then substitute (\ref{discrest}) into the integral (\ref{integral}) and split the integration interval into two equal parts,
\begin{displaymath}
I_n \geq \frac{2}{\pi} \frac{1}{\l_{2N}-\l_1} \int_{\l_{2n}}^{\tau_n} \frac{\sqrt{\l-\l_{2n}}\sqrt{\l_{2n+1}-\l}}{\l_{2N}-\l_1} \, d\l,
\end{displaymath}
where $\tau_n= (\l_{2n}+\l_{2n+1})/2$. For $\l_{2n} \leq \l \leq \tau_n$ we estimate the quantity $\l_{2n+1}-\l$ from below by $\gamma_n/2$, yielding
\begin{displaymath}
I_n \geq \frac{2}{\pi} \frac{1}{\l_{2N}-\l_1} \int_{\l_{2n}}^{\tau_n} \sqrt{\frac{\gamma_n}{2}} \sqrt{\l-\l_{2n}} \; d\l = \frac{1}{3\pi(\l_{2N}-\l_1)} \; \gamma_n^2.
\end{displaymath}

It remains to verify (\ref{discrest}). 
Recall that $\l_{2n}$ and $\l_{2n+1}$ are either both periodic or both antiperiodic eigenvalues of $L$. 
If $\l_{2n}$ and $\l_{2n+1}$ are periodic eigenvalues, we have $\Delta(\l) \geq 2$ for $\l_{2n} \leq \l \leq \l_{2n+1}$, i.e. $|\Delta(\l)| = \Delta(\l)$. In order to make writing easier, let us assume that $N$ is even - the case where $N$ is odd is treated in the same way. Then by (\ref{deltalambdapm2}), $\l_1$ and $\l_{2N}$ are periodic eigenvalues of $L$ and thus for any $\l_{2n} \leq \l \leq \l_{2n+1}$, the left side of (\ref{discrest}) can be estimated from below by
\begin{displaymath}
\frac{\sqrt{\Delta(\l)-2}}{\sqrt{\Delta(\l)+2}} = \sqrt{\frac{(\l-\l_{2n}) (\l-\l_{2n+1})} {(\l-\l_2) (\l-\l_{2N-1})}} \cdot R \geq \frac{\sqrt{\l-\l_{2n}}\sqrt{\l_{2n+1}-\l}}{\l_{2N}-\l_1} \cdot R,
\end{displaymath}
where
\begin{equation} \label{rest}
R \equiv R(\l) = \sqrt{\frac{\l-\l_1}{\l-\l_3} \cdots \frac{\l-\l_{2n-3}}{\l-\l_{2n-1}} \frac{\l_{2n+4}-\l}{\l_{2n+2}-\l} \cdots \frac{\l_{2N}-\l}{\l_{2N-2}-\l}} \, .
\end{equation}
As each of the the fractions under the square root in (\ref{rest}) can be estimated from below by $1$, for any $\l_{2n} \leq \l \leq \l_{2n+1}$ it follows that $R(\l) \geq 1$ on $[\l_{2n}, \l_{2n+1}]$, leading to the claimed estimate (\ref{discrest}).
\end{proof}

\section{Proof of Theorem \ref{ingammantotest}} \label{prooftotalest}

In this Appendix we prove Theorem \ref{ingammantotest} using estimates derived in \cite{bbgk}. Let $(b,a)$ be in $\Mmba$ with $\b \in \R$ and $\a > 0$ arbitrary.

To show Theorem \ref{ingammantotest} we need the following
\begin{prop} \label{harmcomplprop}
For any $(b,a) \in \Mmba$ with $\b \in \R$, $\a > 0$ arbitrary and any $1 \leq n \leq N$,
\begin{equation} \label{esta4}
\l_{2n}(b,a) - \l_{2n-1}(b,a) \leq \frac{2 \pi \alpha}{N}.
\end{equation}
\end{prop}

Before proving Proposition \ref{harmcomplprop} we show how to use it to prove Theorem \ref{ingammantotest}.

\begin{proof}[Proof of Theorem \ref{ingammantotest}]
We begin by adding up the inequalities (\ref{estingnorig}) and get
\begin{equation} \label{sumlemmasing}
\sum_{n=1}^{N-1} \gamma_n^2 \leq 3 \pi \, (\l_{2N} - \l_1) \left( \sum_{n=1}^{N-1} I_n \right).
\end{equation}
Note that
\begin{displaymath}
  \l_{2N} - \l_1 = \sum_{n=1}^{N-1} \g_n + \sum_{n=1}^N (\l_{2n} - \l_{2n-1}).
\end{displaymath}
By the estimate of Proposition \ref{harmcomplprop} we get for any $(b,a) \in \Mmba$
\begin{displaymath}
\l_{2N} - \l_1 \leq 2 \pi \alpha + \sum_{n=1}^{N-1} \gamma_n
\end{displaymath}
which we substitute into (\ref{sumlemmasing}) to yield
\begin{displaymath}
\sum_{n=1}^{N-1} \gamma_n^2 \leq 6 \alpha \pi^2 \, \left( \sum_{n=1}^{N-1} I_n \right) + 3 \pi \, \left( \sum_{n=1}^{N-1} \gamma_n \right) \left( \sum_{n=1}^{N-1} I_n \right).
\end{displaymath}
Using the inequality
\begin{displaymath}
  2ab \leq \epsilon^2 a^2 + \frac{1}{\epsilon^2} b^2 \quad (a,b \in \R, \epsilon > 0)
\end{displaymath}
with $a = \sum_{n=1}^{N-1} \gamma_n$, $b = \sum_{n=1}^{N-1} I_n$, and $\epsilon^2 = \frac{1}{3\pi(N-1)}$, one gets
\begin{displaymath} 
  \sum_{n=1}^{N-1} \gamma_n^2 \leq 6 \pi^2 \alpha \! \left( \sum_{n=1}^{N-1} I_n \! \right) + \frac{3\pi}{2} \! \left( \frac{1}{3\pi(N \!\!-\!\! 1)} \! \left( \sum_{n=1}^{N-1} \gamma_n \! \right)^2 \!\! + 3\pi (N \!\!-\!\! 1) \left( \sum_{n=1}^{N-1} I_n \! \right)^2 \right).
\end{displaymath}
As $\left( \sum_{n=1}^{N-1} \gamma_n \right)^2 \leq (N-1) \left( \sum_{n=1}^{N-1} \gamma_n^2 \right)$, one then concludes that
\begin{equation} \label{lastestimate}
  \frac{1}{2} \sum_{n=1}^{N-1} \gamma_n^2 \leq 6 \pi^2 \alpha \, \left( \sum_{n=1}^{N-1} I_n \right) + \frac{9\pi^2}{2} (N-1) \left( \sum_{n=1}^{N-1} I_n \right)^2,
\end{equation}
which is the claimed estimate (\ref{estingnsumorig}). 
\end{proof}

To prove Proposition \ref{harmcomplprop} we first need to make some preparations. Note that for an element of the form $(b,a) = (\b 1_N, \a 1_N)$ one has, by Lemma \ref{speclemma}, 
\setlength\arraycolsep{2pt}{
\begin{eqnarray*}
\l_{2n}(\b 1_N, \a 1_N) - \l_{2n-1}(\b 1_N, \a 1_N) & = & 2 \a \Big( \cos \frac{(n-1)\pi}{N} - \cos \frac{n\pi}{N} \Big) \\
& = & 4 \a \sin \frac{(2n-1)\pi}{2N} \sin \frac{\pi}{2N} \\
 & < & \frac{2 \pi \a}{N}.
\end{eqnarray*}}
Hence to prove Proposition \ref{harmcomplprop} it suffices to show that for any $(b,a) \in \Mmba$ and any $1 \leq n \leq N$
\begin{equation}
  \l_{2n}(b,a) - \l_{2n-1}(b,a) \leq \l_{2n}(-\b 1_N, \a 1_N) - \l_{2n-1}(-\b 1_N, \a 1_N).
\end{equation}

To this end, following \cite{maos} (cf. also \cite{gatr1}), we introduce the conformal map
\begin{equation} \label{deltamapdefine}
  \delta(\l) := (-1)^N \int_{\l_1}^{\l} \frac{\dot{\Delta}(\mu)}{\sqrt{4 - \Delta^2(\mu)}} \, d\mu,
\end{equation}
where the sign of the square root is chosen such that for $\mu < \l_1$, $\sqrt{4 - \Delta^2(\mu)}$ has positive imaginary part. It is defined on the upper half plane $U := \{ \textrm{Im} \, z > 0 \}$ and its image is the spike domain
$$ \Omega(b,a) := \{ x + iy : 0 < x < N \pi, \, y > 0 \} \setminus \bigcup_{n=1}^{N-1} T_n $$
where for $1 \leq n \leq N-1$, $T_n$ denotes the spike
\begin{displaymath}
T_n := \left\{ n\pi + i t: 0 < t \leq \; \textrm{arcosh} \; \left( (-1)^{N+n} \frac{\Delta(\dot{\l}_n)}{2} \right) \right\}.
\end{displaymath}

To see that $\delta(U) = \Omega(b,a)$, note that for any $(b,a) \in \M$ and $\l \in U$ the discriminant $\Delta(\l)$ and the function $\d(\l)$ are related by the formula
\begin{equation} \label{discrdeltaformula}
  \Delta(\l) = 2(-1)^N \cos \d(\l).
\end{equation}

To prove (\ref{discrdeltaformula}), recall that for $-1<t<1$, one has $\frac{d}{dt} \arccos t = \frac{1}{\sqrt[+]{1 - t^2}}$. This formula remains valid for any $t$ in $\C \setminus ((-\infty,-1] \cup [1,\infty))$. Thus
\begin{displaymath}
    \d(\l) = (-1)^N \int_{\l_1}^{\l} \frac{\dot{\Delta}(\mu)}{\sqrt{4 - \Delta^2(\mu)}} \, d\mu = \int_{\l_1}^{\l} \frac{d}{d \mu} \arccos \left( (-1)^N \frac{\Delta(\mu)}{2} \right).
\end{displaymath}
Since by (\ref{deltalambdapm2}), $\Delta(\l_1) = 2 (-1)^N$, we then get
\begin{equation} \label{deltadiscr}
 \d(\l) = \arccos \left( (-1)^N \frac{\Delta(\mu)}{2} \right) \Big|_{\l_1}^\l = \arccos \left( (-1)^N \frac{\Delta(\l)}{2} \right),
\end{equation}
leading to formula (\ref{discrdeltaformula}) and the claimed statement that $\delta(U) = \Omega(b,a)$.

The map $\delta$ can be extended continuously to the closure $\{ \textrm{Im} \, z \geq 0 \}$ of the upper half plane. This extension, again denoted by $\delta$, is 2-1 over each nontrivial spike $T_n$ and 1-1 otherwise. Since the $n$-th spike $T_n$ is the image under $\delta$ of the $n$-th gap $(\l_{2n}, \l_{2n+1})$, all spikes are empty iff all gaps are collapsed. By Lemma \ref{speclemma}, all gaps are collapsed for $(b,a) = (-\b 1_N, \a 1_N)$, hence $\Omega(b,a) \subset \Omega(-\b 1_N, \a 1_N)$ for any $(b,a) \in \Mmba$.

Note that
\begin{equation} \label{bandintrepr}
  \l_{2n}(b,a) - \l_{2n-1}(b,a) = \delta^{-1} (n\pi-) - \delta^{-1} ((n-1) \pi+) = \int_{-\infty}^\infty u^{(n)}(\delta(\l)) d \, \l,
\end{equation}
where $u^{(n)}: \Omega(b,a) \to \R, \, z \mapsto u^{(n)}(z;b,a)$ is the harmonic measure of the open subset $((n-1) \pi, n \pi)$ of $\partial \Omega(b,a)$ (see e.g. \cite{garnett} for the notion of the harmonic measure).

We need two lemmas from complex and harmonic analysis, respectively.

\begin{lemma} \label{complanalemma}
For $(b,a) = (-\b 1_N, \a 1_N)$ with $\b \in \R$, $\a > 0$ arbitrary, the map $\d(\l)$ defined by (\ref{deltamapdefine}) is given by
\begin{equation} \label{deltastandard}
  \delta(\l) = N \arccos \left( -\frac{\l + \b}{2\a} \right).
\end{equation}
For arbitrary $(b,a)$ in $\Mmba$ and $\xi \in \R$, the following asymptotic estimate holds as $\eta \to \infty$
\begin{equation} \label{deltageneral}
  \delta(\xi + i \eta) = N \arccos \left( -\frac{\xi + i \eta + \b}{2\a} \right) + O(\eta^{-2}), 
\end{equation}
locally uniformly in $\xi$.
\end{lemma}

\begin{proof}[Proof of Lemma \ref{complanalemma}]
In view of the formulas (\ref{y1betaalpha}) and (\ref{y2betaalpha}) for the fundamental solutions $y_1$ and $y_2$ for $(b,a) = (-\b 1_N, \a 1_N)$, the discriminant $\Delta(\l) = \Delta(\l, -\b 1_N, \a 1_N)$ is given by
\begin{eqnarray*}
  \Delta(\l) & = & y_1(N,\l) + y_2(N+1,\l) \\
& = & -\frac{\sin (\rho (N-1))}{\sin \rho} + \frac{\sin (\rho (N+1))}{\sin \rho} \\
& = & 2 \cos (\rho N),
\end{eqnarray*}
where $\pi < \rho < 2\pi$ is determined by $\cos \rho = \frac{\l + \b}{2\a}$. Hence
\begin{equation} \label{deltaspecialdiscr}
  \Delta(\l) = 2 T_N \left( \frac{\l + \b}{2\a} \right)
\end{equation}
where for any $z \in U$, 
\begin{equation} \label{chebcompl}
  T_N(z) = \cos (N \arccos z).
\end{equation}
Actually, $T_N(z)$ is a polynomial in $z$ of degree $N$, referred to as Chebychev polynomial of the first kind. 
Substituting (\ref{deltaspecialdiscr}) into (\ref{deltadiscr}), we obtain
\begin{displaymath}
  \d(\l) = \arccos \left( (-1)^N \frac{\Delta(\l)}{2} \right) = \arccos \left( (-1)^NT_N \left( \frac{\l + \b}{2\a} \right) \right).
\end{displaymath}
The claimed identity (\ref{deltastandard}) now follows from the elementary symmetry
\begin{displaymath}
T_N(z) = (-1)^N T_N(-z) \quad \forall \, z \in \C.
\end{displaymath}
Now let $(b,a) \in \Mmba$. The asymptotic estimate (\ref{deltageneral}) follows by comparing the polynomials $\Delta(\l)$ correponding to $(b,a)$ and the one corresponding to $(-\b 1_N, \a 1_N)$. By (\ref{discrdef}) (and the discussion following it), in both cases,
\begin{displaymath}
  \Delta(\l) = \a^{-N} \l^N \left( 1 + N \b \l^{-1} + O(\l^{-2}) \right) \quad \textrm{as} \quad |\l| \to \infty.
\end{displaymath}
This implies that
\begin{displaymath}
\Delta_{b,a}(\l) = \Delta_{-\b 1_N, \a 1_N}(\l) \cdot (1 + O(\l^{-2})),
\end{displaymath}
hence by (\ref{deltastandard}) and (\ref{deltaspecialdiscr})
\begin{eqnarray*}
  \delta_{b,a}(\l) & = & \arccos \left( (-1)^N \frac{\Delta_{b,a}(\l)}{2} \right) \\
& = & \arccos \left( \frac{(-1)^N}{2} \Delta_{-\b 1_N, \a 1_N}(\l) \cdot (1 + O(\l^{-2})) \right). \\
& = & \arccos \left( (-1)^N T_N \left( \frac{\l + \b}{2\a} \right) \cdot (1 + O(\l^{-2})) \right) \\
& = & \arccos \left( T_N \left( -\frac{\l + \b}{2\a} \right) \cdot (1 + O(\l^{-2})) \right).
\end{eqnarray*}
Substituting formula (\ref{chebcompl}) for $T_N$, one then concludes that
\begin{eqnarray}
\delta_{b,a}(\l) & = & \arccos \left( \cos \left( N \arccos \left( -\frac{\l + \b}{2\a} \right) \right) \cdot (1 + O(\l^{-2})) \right) \nonumber\\
& = & N \arccos \left( -\frac{\l + \b}{2\a} \right)+ O(\l^{-2}), \label{deltabaarccosasympt}
\end{eqnarray}
where in the last step we used that $\arccos z = -i\log (z + i\sqrt[+]{1-z^2})$ for any $z$ in $\C \setminus ((-\infty,-1] \cup [1,\infty))$.
\end{proof}

\begin{lemma} \label{harmanalemma}
Let $u: \Omega(b,a) \to \R$ be a bounded harmonic function such that the nontangential limit of $u(z)$ on $\partial \Omega(b,a)$ has compact support, and let $U(\l) := u(\d(\l))$, where $\d(\l)$ is the function defined by (\ref{deltamapdefine}). Then for almost every $t \in \R$, the limit $U(t) := \lim_{\eta \to 0} U(t + i\eta)$ exists and is integrable, and
\begin{equation} \label{sinhformula}
  \int_{-\infty}^\infty U(t) dt = \lim_{x \to \infty} 2 \pi \a \sinh \left( \frac{x}{N} \right) u \left( \frac{N\pi}{2} + i x \right).
\end{equation}
\end{lemma}

\begin{proof}[Proof of Lemma \ref{harmanalemma}]
Again by Fatou's theorem, for a.e. $t$ the (nontangential) limit $\lim_{\eta \to 0} U(t + i\eta)$ exists, since $U$ is a bounded harmonic function on $\{ \textrm{Im} \; (z) > 0 \}$. Since $u$ is bounded on $\Omega(b,a)$ and its nontangential limit to $\partial \Omega$ has compact support, $U(t)$ is bounded and compactly supported and thus in particular integrable. For $\l = \xi + i \eta$, one then has the Poisson representation
\begin{equation} \label{uprepr}
  U(\xi + i \eta) = \frac{\eta}{\pi} \int_{-\infty}^{\infty} \frac{U(t)}{(t - \xi)^2 + \eta^2},
\end{equation}
and by dominated convergence we conclude that
\begin{equation} \label{limUtdt}
  \int_{-\infty}^{\infty} U(t) \, dt = \lim_{\eta \to \infty} \pi \eta \, U(\xi + i \eta).
\end{equation}
(In particular, the limit in the latter expression exists.)
In order to compute the right hand side of (\ref{limUtdt}), let $\xi + i \eta$ be given by
\begin{displaymath}
  \xi + i \eta = \d^{-1}\left( \frac{N \pi}{2} + i x \right),
\end{displaymath}
for $x$ sufficiently large. Then $U(\xi + i \eta) = u \left( \frac{N \pi}{2} + i x \right)$, and by (\ref{deltabaarccosasympt}), it follows that
\begin{equation} \label{harmsubstansatz}
 \frac{\pi}{2} + \frac{i}{N} x = \arccos \left( -\frac{(\xi + \b) + i \eta}{2 \a} \right) + O\left( (\xi + i \eta)^{-2} \right) \quad (x \to \infty). 
\end{equation}
Taking the cosine of both sides of (\ref{harmsubstansatz}), multiplying by $-2\a$ and using that $\cos(\frac{\pi}{2} + it) = -i \sinh t$ for $t \in \R$, we obtain
\begin{displaymath}
  2i\a \sinh \frac{x}{N} = ((\xi + \b) + i \eta) \left( 1 + O\left( (\xi + i \eta)^{-2} \right) \right) \quad (x \to \infty).
\end{displaymath}
Hence, as $x \to \infty$,
\begin{equation} \label{xietaresult}
  \xi = O(1), \quad \eta = 2\a \, \sinh \frac{x}{N} + O(1).
\end{equation}
Substituting (\ref{xietaresult}) into (\ref{limUtdt}) leads to the claimed formula (\ref{sinhformula}).
\end{proof}

\begin{proof}[Proof of Proposition \ref{harmcomplprop}]
Let $(b,a) \in \Mmba$. Besides the harmonic measure $u^{(n)}$ of the set $E := ((n-1) \pi, n \pi) \subset \partial \Omega(b,a)$ we also consider the harmonic measure $u_{\b,\a}^{(n)}$ of $E \subset \partial \Omega(-\b 1_N, \a 1_N)$; note that
\begin{displaymath}
\Omega(-\b 1_N, \a 1_N) = \{ x+iy | \, 0 < x < N\pi, \, y > 0 \}
\end{displaymath}
and hence $\Omega(b,a) \subseteq \Omega(-\b 1_N, \a 1_N)$. According to \cite{garnett}, both $u^{(n)}$ and $u_{\b,\a}^{(n)}$ satisfy the hypotheses of Lemma \ref{harmanalemma}. 
Let us recall the monotonicity property of the harmonic measures $u(z, E, \Omega)$ with respect to $\Omega$ (see e.g. \cite{garnett}): If $\Omega_1 \subseteq \Omega_2$, $E \subset \partial \Omega_1 \cap\partial \Omega_2$, and $u(z, E, \Omega_i)$ ($i=1,2$) denotes the harmonic measure of $E \subset \partial \Omega_i$, then for any $z \in \Omega_1$, $u(z,E,\Omega_1) \leq u(z,E,\Omega_2)$. Apply this general principle to $\Omega_1 :=  \Omega(b,a)$ and $\Omega_2 := \Omega(-\b 1_N, \a 1_N)$ to get
\begin{equation} \label{maxprincest}
  u^{(n)} (x) \leq u_{\b,\a}^{(n)} (x).
\end{equation}

Writing $U^{(n)}(\l) := u^{(n)}(\d(\l))$ as well as $U_{\b,\a}^{(n)}(\l) := u_{\b,\a}^{(n)}(\d(\l))$ and combining (\ref{bandintrepr}), (\ref{sinhformula}), and (\ref{maxprincest}), we conclude that
\begin{eqnarray*}
  \l_{2n}(b,a) - \l_{2n-1}(b,a) & = & \int_{-\infty}^\infty U^{(n)}(\l) d \, \l \\
& = & \lim_{x \to \infty} 2 \pi \a \sinh \left( \frac{x}{N} \right) u^{(n)} \left( \frac{N\pi}{2} + i x \right) \\
& \leq & \lim_{x \to \infty} 2 \pi \a \sinh \left( \frac{x}{N} \right) u_{\b,\a}^{(n)} \left( \frac{N\pi}{2} + i x \right) \\
& = & \int_{-\infty}^\infty U_{\b,\a}^{(n)}(\l) d \, \l \\
& = & \l_{2n}(-\b 1_N, \a 1_N) - \l_{2n-1}(-\b 1_N, \a 1_N).
\end{eqnarray*}
This completes the proof of the estimate (\ref{esta4}) and therefore of Proposition \ref{harmcomplprop}.
\end{proof}

\vspace{.4cm}

\textsc{Institut f\"ur Mathematik, Universit\"at Z\"urich, Winterthurerstrasse 190, CH-8057 Z\"urich, Switzerland} \\
\emph{E-mail address:} \texttt{andreas.henrici@math.unizh.ch}

\vspace{.4cm}

\textsc{Institut f\"ur Mathematik, Universit\"at Z\"urich, Winterthurerstrasse 190, CH-8057 Z\"urich, Switzerland} \\
\emph{E-mail address:} \texttt{thomas.kappeler@math.unizh.ch}


\begin{thebibliography}{99}

\bibitem{bbgk}
\textsc{D. B\"attig, A. M. Bloch, J. C. Guillot \& T. Kappeler}, On the
symplectic structure of the phase space for periodic KdV, Toda,
and defocusing NLS. \emph{Duke Math. J.} \textbf{79} (1995),
549-604.

\bibitem{bggk}
\textsc{D. B\"attig, B. Gr\'ebert, J. C. Guillot \& T. Kappeler},
Fibration of the phase space of the periodic Toda lattice.
\emph{J. Math. Pures Appl.} \textbf{72} (1993), 553-565.

\bibitem{fla1}
\textsc{H. Flaschka}, The Toda lattice. I. Existence of integrals.
\emph{Phys. Rev.}, Sect. B \textbf{9} (1974), 1924-1925.

\bibitem{flmc}
\textsc{H. Flaschka \& D. McLaughlin}, Canonically conjugate
variables for the Korteweg-de Vries equation and the Toda lattice
with periodic boundary conditions. \emph{Prog. Theor. Phys.}
\textbf{55} (1976), 438-456.

\bibitem{fpu}
\textsc{E. Fermi, J. Pasta \& S. Ulam}, Studies of non linear problems. \emph{Los Alamos Rpt.} \textbf{LA-1940} (1955). In: \emph{Collected Papers of Enrico Fermi}. University of Chicago Press, Chicago, 1965, Volume II, 978-988. Theory, Methods and Applications, 2nd ed., Marcel Dekker, New York, 2000.

\bibitem{garnett}
\textsc{J. Garnett}, \emph{Applications of Harmonic Measure}. University of Arkansas Lecture Notes in Math. \textbf{8}, Wiley, New York, 1986.

\bibitem{gatr1}
\textsc{J. Garnett \& E. Trubowitz}, Gaps and bands of one dimensional periodic Schr\"odinger operators. \emph{Comm. Math. Helv.} \textbf{59} (1984), 258-312.

\bibitem{gkp}
\textsc{B. Gr\'ebert, T. Kappeler \& J. P\"oschel}, Normal form theory for the NLS equation: a preliminary report. Preprint, 2003.

\bibitem{ahtk2}
\textsc{A. Henrici \& T. Kappeler}, Global Birkhoff coordinates for the periodic Toda lattice. Preprint, 2008.

\bibitem{ahtk3}
\textsc{A. Henrici \& T. Kappeler}, Birkhoff normal form for the periodic Toda lattice. \texttt{arXiv:nlin/0609045v1 [nlin.SI]}. To appear in \emph{Contemp. Math.}

\bibitem{kama}
\textsc{T. Kappeler \& M. Makarov}, On Birkhoff coordinates for KdV.
\emph{Ann. Henri Poincar\'{e}} \textbf{2} (2001), 807-856.

\bibitem{kapo}
\textsc{T. Kappeler \& J. P\"oschel}, \emph{KdV~\&~KAM}. Ergebnisse der Mathematik, 3. Folge, vol. \textbf{45}, Springer, 2003.

\bibitem{kato}
\textsc{T. Kappeler \& P. Topalov}, Global Well-Posedness of KdV in $H^{-1}(\T, \R)$. \emph{Duke Math. J.} \textbf{135}(2) (2006), 327-360.

\bibitem{mana}
\textsc{S. V. Manakov}, Complete integrability and stochastization of discrete dynamical systems. \emph{Zh. Exp. Teor. Fiz.} \textbf{67} (1974), 543-555 [Russian]. English translation: \emph{Sov. Phys. JETP} \textbf{40} (1975), 269-274.

\bibitem{maos}
\textsc{V. A. Marchenko \& I. V. Ostrovskii}, A characterization of the spectrum of Hill's operator, \emph{Math. SSSR-Sbornik} \textbf{97} (1975), 493-554.

\bibitem{mcva1}
\textsc{H. P. McKean \& K. L. Vaninsky}, Action-angle variables for
the cubic Schroedinger equation. \emph{Comm. Pure Appl. Math.}
\textbf{50} (1997), 489-562.

\bibitem{moer}
\textsc{P. van Moerbeke}, The spectrum of Jacobi matrices.
\emph{Invent. Math.} \textbf{37} (1976), 45-81.

\bibitem{teschl2}
\textsc{G. Teschl}, \emph{Jacobi Operators and Completely
Integrable Nonlinear Lattices}. Math. Surveys and Monographs
\textbf{72}, Amer. Math. Soc., Providence, 2000.

\bibitem{toda}
\textsc{M. Toda}, \emph{Theory of Nonlinear Lattices}, 2nd enl.
ed., Springer Series in Solid-State Sciences \textbf{20},
Springer, Berlin, 1989.

\bibitem{tsuji}
\textsc{M. Tsuji}, \emph{Potential Theory in Modern Function Theory}, Mruzen, Tokyo, 1959.



\end{thebibliography}
\end{document}